\documentclass[11pt]{article}

\usepackage{amsfonts,amsmath,amssymb,amsthm}
\usepackage{indentfirst,color}
\usepackage[a4paper,ignoreall]{geometry}
\usepackage{hyperref,enumerate}
\usepackage{authblk}
\usepackage{graphicx, xypic, tikz}
\usepackage{algorithm}
\usepackage{algorithmicx}
\usepackage{algpseudocode}
\usepackage{verbatim}
\usepackage{url,hyperref,enumerate}
\usepackage{mathtools}
\usepackage{tikz}
\usepackage{subcaption}
\usepackage{caption}
\usepackage{float}
\usepackage[T1]{fontenc}

\newcommand{\tr}{ \rm Tr}
\newcommand{\Tr}{ \rm Tr}
\newcommand{\gf}{ \mathbb{F}}
 
\newcommand{\Zz}{\mathbb {Z}}
\newcommand{\Z}{\mathbb {Z}}
\newcommand{\calD}{ \mathcal D}

\newcommand{\fqm}{ {\gf_{q^m}} }
\newcommand{\fq}{ {\gf_{q}} }
\newcommand{\fqmstar}{ {\gf_{q^m}^*} }
\newcommand{\fqstar}{ {\gf_{q}^*} }

\newtheorem{theorem}{Theorem}
\newtheorem{corollary}{Corollary}

\newtheorem{lemma}{Lemma}

\newtheorem{example}{Example}

\newtheorem{remark}{Remark}

\def\yin#1 {\fbox {\footnote {\ }}\ \footnotetext { From Yin: {\color{blue}#1}}}
\def\hyin#1 {}

\makeatletter
\newcommand{\figcaption}{\def\@captype{figure}\caption}
\newcommand{\tabcaption}{\def\@captype{table}\caption}
\makeatother

\begin{document}

\title{Six Constructions of Difference Families}
\author[1]{Cunsheng Ding, \thanks{Email: cding@ust.hk}}
\author[2]{Chik How Tan, \thanks{Email: tsltch@nus.edu.sg}}
\author[3]{Yin Tan, \thanks{Corresponding author. Email: yin.tan@uwaterloo.ca. Tel: +1 (519) 888-4567 EXT 31246.}}

\affil[1]{Department of Computer Science and Engineering \authorcr The 
Hong Kong University of Science and Technology\authorcr Clear Water Bay,
Kowloon, Hong Kong}
\affil[2]{Temasek Laboratories\authorcr  National University of Singapore, 117411 Singapore}
\affil[3]{Department of Electrical and Computer Engineering\authorcr University of Waterloo, Canada}

\renewcommand\Authands{ and }

\maketitle

\begin{abstract}
	In this paper, six constructions of difference families are presented. These constructions make 
use of difference sets, almost difference sets and disjoint difference families, and give new point 
of views of relationships among these combinatorial objects. Most of the constructions work for
all finite groups. Though these constructions look simple, they produce many difference families 
with new parameters. In addition to the six new constructions, new results about intersection 
numbers are also derived. 
  

\end{abstract}

\section{Introduction}

Let $G$ be a group of order $v$. A $(v, K, \gamma; u)$ \textit{difference family} (DF) in $G$ is 
a collection of subsets of $G$, $\calD=\{D_1, D_2, \ldots,D_u\}$, such that the multiset union
\begin{eqnarray*}
  \sum_{i=1}^u \{ xy^{-1}: x, y \in D_i, x \neq y \} = \gamma(G \setminus \{1\}),
\end{eqnarray*}
where $K$ is the set of cardinalities of all the base blocks $D_i$ and $1$ is the identity element of $G$. In the case of
$K=\{k\}$, it is called a $(v, k, \gamma; u)$ difference family.
If these $D_i$'s are pairwise disjoint, $\calD$ is called a {\em disjoint difference
family} (DDF). If a $(v, k, \gamma; u)$ difference family exists, we have
\begin{eqnarray}
uk(k-1)=\gamma (v-1).
\end{eqnarray}
Our definition of difference families here is taken from the monograph \cite{BJL}.

Difference families are well studied \cite{Bura,Cding,CD05,DR,DS,Miao,Mart,Wilson}.
A lot of information on difference families can be found in the monograph \cite{BJL}.
Difference families have applications in coding theory and cryptography \cite{Ogata}.
For various applications, it is necessary to develop constructions of difference families 
with flexible parameters.

Let $G$ be a group of order $v$ and $D$ a $k$-subset of
$G$. The set $D$ is a $(v, k, \lambda)$ {\em difference set} (DS) in $G$ if the differences
$$ d_1d_2^{-1},\ d_1,d_2\in D,d_1\ne d_2 $$
represent each nonidentity element of $G$ exactly $\lambda$ times.
We call $D$ an {\em Abelian (\emph{resp.} cyclic) difference set} if $G$ is an Abelian (resp. cyclic) group. 
Notice that $D$ is a $(v,k,\lambda)$ difference set if and only if 
$|D\cap Dw|=\lambda$ for all nonidentity element $w$ of $G$.
It follows from
\cite[Corollary 1.8]{BJL} that the form $d_1^{-1}d_2, d_1,d_2\in D$ may also represent each
nonidentity element $g$ exactly $\lambda$ times. The reader is referred to \cite{BJL} for 
further information on difference sets.

In this paper, we present four constructions (Theorems \ref{dsdf}, \ref{construction2}, \ref{BIBDconstruction}, 
\ref{thm-mainmain}) of difference families using difference sets.
Three of them work for all finite groups, while one requires that $G$ be Abelian. 
In Section \ref{sec-fourfour}, in order to obtain more general results, we use balanced incomplete block designs 
(defined in Section \ref{sec-fourfour}) 
to construct difference families (Theorem \ref{BIBDconstruction}). 
As an application, a new construction (Corollaries \ref{construction3}, \ref{construction5}) of difference families using difference sets is obtained. 
In Section \ref{sec-fivefive}, we present a construction (Theorem \ref{thm-mainmain}) using Abelian difference sets, and 
determine the intersection numbers (defined in Section \ref{sec-fivefive}) for several cases.
Moreover, we use Gauss sums to compute the intersection numbers for general cases of the Singer difference sets (Theorem \ref{semiprimitive}). It is noted that a particular case of the intersection number problem is solved in \cite{kestenband} using geometric arguments and this problem is also considered in a recent paper \cite{momihara}.

In Section \ref{sec-disj}, we first use near-resolvable block designs to construct difference families 
and then apply the construction to disjoint difference families to obtain a specific class of difference 
families.

Let $G$ be an Abelian group of order $v$. A $k$-subset $D$ of $G$
is a $(v, k, \lambda, t)$ {\em almost difference set} (ADS) in $G$ if
the difference function $d_D(w)=|D \cap Dw|$ takes on $\lambda$
altogether $t$ times and $\lambda +1$ altogether $v-1-t$ times when $w$
ranges over all the nonidentity elements of $G$. The reader is referred to
\cite{ADHKM} for a survey of almost difference sets. 
In Section \ref{sec-adscons}, we give a construction of difference families employing 
almost difference families (Theorem \ref{adsdf}). We make concluding remarks in Section 
\ref{sec-final}.

\section{The first construction with difference sets}\label{sec-twotwo}

In this section, we give the first construction of difference families using difference sets.

\begin{theorem}
  \label{dsdf}
  Let $D$ be a $(v,k,\lambda)$ difference set in a group $G$. For any nonidentity 
  element $x\in G$, define the set 
  $$
   D_x=D\cap Dx,
  $$
  where $Dx:=\{dx: d \in D\}$.
  Then $\mathcal{F}=\{ D_x: x\in G\setminus\{1\}\}$ is a $(v,\lambda,\lambda(\lambda-1); v-1)$ 
  difference family in G. 
\end{theorem}
\begin{proof}
Let $G^*=G\setminus\{1\}$. Since $D$ is a $(v,k,\lambda)$ difference set,
  $|D_x|=\lambda$ for every $x \in G^*$.
  Now we need to compute the multiset $M:=\{ d_1^{-1}d_2: d_1,d_2\in D_x, x\in G^* | d_1\ne d_2 \}$.
  For each $g\in G^*$, there are $\lambda$ pairs $(d_1,d_2)\in D\times D$
  such that $g=d_1^{-1}d_2$. Let $(d_1,d_2)$ be such a pair, where $d_1 \in D$ and
  $d_2 \in D$. We compute the multiplicity of $g=d_1^{-1}d_2$ in $M$.
  For this fixed pair $(d_1, d_2)$, the number of $x \in G^*$ such that
  $$
   d_1 \in Dx \mbox{ and } d_2 \in Dx
  $$
  is equal to
  $$
   |Dd_1^{-1}\cap Dd_2^{-1}|-1= |D \cap Dd_2^{-1}d_1|-1=\lambda-1.
  $$
  Hence, the multiplicity of $g$ in $M$ is $\lambda(\lambda-1)$. This is true
  for every $g \in G^*$. The conclusion of this theorem then follows.   
\end{proof}

Applying Theorem \ref{dsdf} to known Abelian difference sets, we obtain
difference families with new parameters described in the following corollaries.

\begin{corollary}\label{cor-11}
Let $m\ge 3$ be an integer, and let $q$ be a power of a prime. If $D$ is the Singer difference set in
$(\gf_{q^m}^*/\gf_{q}^*, \times)$ with parameters \cite{Sing}
$$
\left(\frac{q^m-1}{q-1}, \frac{q^{m-1}-1}{q-1}, \frac{q^{m-2}-1}{q-1}\right),
$$
then $\{D_x: x \in \gf_{q^m}^*/\gf_{q}^* \setminus \{1\}\}$ is a
difference family in $(\gf_{q^m}^*/\gf_{q}^*, \times)$ with parameters
$$
\left(\frac{q^m-1}{q-1},  \frac{q^{m-2}-1}{q-1}, \frac{q^{m-2}-1}{q-1}\frac{q^{m-2}-q}{q-1}; \frac{q^m-q}{q-1}\right).
$$

If $D$ is the Singer difference set in
$(\gf_{q^m}^*/\gf_{q}^*, \times)$ with parameters \cite{Sing}
$$
\left(\frac{q^m-1}{q-1}, q^{m-1}, q^{m-2}(q-1)\right),
$$
then $\{D_x: x \in \gf_{q^m}^*/\gf_{q}^* \setminus \{1\}\}$ is a
difference family in $(\gf_{q^m}^*/\gf_{q}^*, \times)$  with parameters
$$
\left(\frac{q^m-1}{q-1},  q^{m-2}(q-1), q^{m-2}(q-1)(q^{m-2}(q-1)-1); \frac{q^m-q}{q-1}\right).
$$

\end{corollary}

\begin{corollary}\label{cor-12}
Let $q \equiv 3 \bmod{4}$ be a power of an odd prime. If $D$ is the set of all quadratic residues
in $(\gf_{q}, +)$, then $D$ is a $(q, (q-1)/2, (q-3)/4)$ difference set in
$(\gf_{q}, +)$ \cite{Pale},
and $\{D_x: x \in \gf_{q}^*\}$ is a $(q, (q-3)/4, (q-3)(q-7)/16; q-1)$
difference family in $(\gf_{q}, +)$.

If $D$ consists of all the quadratic nonresidues in $(\gf_{q}, +)$ and $0$, then
$D$ is a $(q, (q+1)/2, (q+1)/4)$ difference set in $(\gf_{q}, +)$,
and $\{D_x: x \in \gf_{q}^*\}$ is a $(q, (q+1)/4, (q+1)(q-3)/16; q-1)$
difference family in $(\gf_{q}, +)$.

\end{corollary}

\begin{corollary}\label{cor-13}
Let $q =4t^2+1$ be a power of a prime, where $t$ is odd. If $D$ is the set of all biquadratic residues
in $(\gf_{q}, +)$, then $D$ is a $(q, (q-1)/4, (q-5)/16)$ difference set in
$(\gf_{q}, +)$ \cite{Chol},
and $\{D_x: x \in \gf_{q}^*\}$ is a $(q, (q-5)/16, (q-5)(q-21)/256; q-1)$
difference family in $(\gf_{q}, +)$.

If $D$ is the complement of the biquadratic-residue difference set, then
$D$ is a $(q, (3q+1)/4, (9q+3)/16)$ difference set in $(\gf_{q}, +)$,
and $\{D_x: x \in \gf_{q}^*\}$ is a $(q, (9q+3)/16, (9q+3)(9q-13)/256; q-1)$
difference family in $(\gf_{q}, +)$.

\end{corollary}

\begin{corollary}\label{cor-14}
Let $q =4t^2+9$ be a power of a prime, where $t$ is odd. If $D$ consists of
all biquadratic residues and zero in $(\gf_{q}, +)$, then $D$ is a
$(q, (q+3)/4, (q+3)/16)$ difference set in
$(\gf_{q}, +)$ \cite{Lehm},
and $\{D_x: x \in \gf_{q}^*\}$ is a $(q, (q+3)/16, (q+3)(q-13)/256; q-1)$
difference family in $(\gf_{q}, +)$.

If $D$ is the set of all biquadratic nonresidues, then
$D$ is a $(q, (3q-3)/4, (9q-21)/16)$ difference set in $(\gf_{q}, +)$,
and $\{D_x: x \in \gf_{q}^*\}$ is a $(q, (9q-21)/16, (9q-21)(9q-37)/256; q-1)$
difference family in $(\gf_{q}, +)$.
\end{corollary}

\begin{corollary}
  \label{cor-15}
  Let $q =8t^2+1=64s^2+9$ be a power of a prime, where $t$ and $s$ are odd.
  If $D$ is the set of all octic residues
  in $(\gf_{q}, +)$, then $D$ is a $(q, (q-1)/8, (q-9)/64)$ difference set in
  $(\gf_{q}, +)$ \cite{Chol},
  and $\{D_x: x \in \gf_{q}^*\}$ is a $(q, (q-9)/64, (q-9)(q-73)/64^2; q-1)$
  difference family in $(\gf_{q}, +)$.

  If $D$ is the complement of the octic-residue difference set, then
  $D$ is a $(q, (7q+1)/8, (7q+1)/64)$ difference set in $(\gf_{q}, +)$,
  and $\{D_x: x \in \gf_{q}^*\}$ is a $(q, (7q+1)/64, (7q+1)(7q-63)/64^2; q-1)$
  difference family in $(\gf_{q}, +)$.
\end{corollary}

\begin{corollary}\label{cor-16}
Let $q =8t^2+49=64s^2+441$ be a power of a prime, where $t$ is odd and $s$
is even.
If $D$ consists of all the octic residues and zero
in $(\gf_{q}, +)$, then $D$ is a $(q, (q+7)/8, (q+7)/64)$ difference set in
$(\gf_{q}, +)$ \cite{Chol},
and $\{D_x: x \in \gf_{q}^*\}$ is a $(q, (q+7)/64, (q+7)(q-57)/64^2; q-1)$
difference family in $(\gf_{q}, +)$.

If $D$ is the set of all the octic nonresidues, then
$D$ is a $(q, (7q-7)/8, (49q-105)/64)$ difference set in $(\gf_{q}, +)$,
and $\{D_x: x \in \gf_{q}^*\}$ is a $(q, (49q-105)/64, (49q-105)(49q-169)/64^2; q-1)$
difference family in $(\gf_{q}, +)$.
\end{corollary}

\begin{corollary}\label{cor-17}
Let $q$ and $q+2$ both be prime powers. If $D$ is the twin-prime-power difference set
in $(\gf_{q} \times \gf(q+2), +)$, then $D$ is a $(q(q+2), (q^2+2q-1)/2, (q^2+2q-3)/4)$
difference set in $(\gf_{q} \times \gf(q+2), +)$ \cite{Stor,Whit},
and
$$\{D_x: x \in \gf_{q}\times \gf(q+2) \setminus \{(0,0)\}\}$$
is a $(q(q+2), (q^2+2q-3)/4, (q^2+2q-3)(q^2+2q-7)/16; q^2+2q-1)$
difference family in $(\gf_{q} \times \gf(q+2), +)$.

If $D$ is the complement of the twin-prime-power difference set, then $D$ is a
$(q(q+2), (q^2+2q+1)/2, (q^2+2q+1)/4)$
difference set in $(\gf_{q} \times \gf(q+2), +)$ \cite{Whit},
and
$$\{D_x: x \in \gf_{q}\times \gf(q+2) \setminus \{(0,0)\}\}$$
is a $(q(q+2), (q^2+2q+1)/4, (q^2+2q+1)(q^2+2q-3)/16; q^2+2q-1)$
difference family in $(\gf_{q} \times \gf(q+2), +)$.
\end{corollary}

The following families of difference sets yield automatically difference
families with new parameters:
\begin{itemize}
\item Hadamard difference sets with parameters $(4s^2, 2s^2-s, s^2-s)$ \cite{BJL}.

\item Chen difference sets with parameters
$$
\left( 4q^{2t} \frac{q^{2t}-1}{q^2-1},  q^{2t-1} \left( \frac{2(q^{2t}-1}{q+1} +1\right),  q^{2t-1}(q-1) \frac{q^{2t-1}+1}{q+1}\right),
$$
where $q$ is a prime power and $t$ is any positive integer \cite{Chen}.
\end{itemize}

Next, we give an example using nonabelian difference sets.
\begin{example}
	In \cite[Example 2.1.8]{Pott}, there is a nonabelian $(100,45,20)$ difference set.
	By Theorem \ref{dsdf}, we can get a $(100,20,380;99)$ difference family in an nonabelian group.
	This example is interesting as Abelian $(100,45,20)$ difference sets do not exist by known results.
\end{example}

Table \ref{tab-111111}  summarizes the new parameters of difference families using Abelian difference sets obtained in this section. 

\begin{table}[ht]
  \label{newparabelian}
\centering
\caption{Parameters of the difference families obtained in Section \ref{sec-twotwo}, where $k$ and $\gamma$ are the block size and frequency respectively.}\label{tab-111111}
\begin{tabular}{|c|c|c|c|c|c|} \hline
$v$ & $k$ & $\gamma$ & u & $G$ & reference \\
\hline \hline
$\frac{q^m-1}{q-1} $ &  $ \frac{q^{m-2}-1}{q-1}$ & $k(k-1)$ & $v-1$ & $(\gf_{q^m}^*/\gf_{q}^*, \times)$ & Cor. \ref{cor-11}\\
\hline
$\frac{q^m-1}{q-1} $ &  $q^{m-2}(q-1) $ & $k(k-1)$ & $v-1$ & $(\gf_{q^m}^*/\gf_{q}^*, \times)$ & Cor. \ref{cor-11}\\
\hline
$q $ &  $\frac{q-3}{4} $ & $k(k-1)$ & $v-1$ & $(\gf_{q}, +)$ & Cor. \ref{cor-12}\\
\hline
$q $ &  $\frac{q+1}{4} $ & $k(k-1)$ & $v-1$ & $(\gf_{q}, +)$ & Cor. \ref{cor-12}\\
\hline
$q $ &  $\frac{q-5}{16} $ & $k(k-1)$ & $v-1$ & $(\gf_{q}, +)$ & Cor. \ref{cor-13}\\
\hline
$q $ &  $\frac{9q+3}{16} $ & $k(k-1)$ & $v-1$ & $(\gf_{q}, +)$ & Cor. \ref{cor-13}\\
\hline
$q $ &  $\frac{q+3}{16} $ & $k(k-1)$ & $v-1$ & $(\gf_{q}, +)$ & Cor. \ref{cor-14}\\
\hline
$q $ &  $\frac{9q-21}{16} $ & $k(k-1)$ & $v-1$ & $(\gf_{q}, +)$ & Cor. \ref{cor-14}\\
\hline
$q $ &  $\frac{q-9}{64} $ & $k(k-1)$ & $v-1$ & $(\gf_{q}, +)$ & Cor. \ref{cor-15}\\
\hline
$q $ &  $\frac{7q+1}{64} $ & $k(k-1)$ & $v-1$ & $(\gf_{q}, +)$ & Cor. \ref{cor-15}\\
\hline
$q $ &  $\frac{q+7}{64} $ & $k(k-1)$ & $v-1$ & $(\gf_{q}, +)$ & Cor. \ref{cor-16}\\
\hline
$q $ &  $\frac{49q-105}{64} $ & $k(k-1)$ & $v-1$ & $(\gf_{q}, +)$ & Cor. \ref{cor-16}\\
\hline
$q(q+2) $ &  $\frac{q^2+2q-3}{4} $ & $k(k-1)$ & $v-1$ & $(\gf_{q} \times \gf_{q+2},+)$ & Cor. \ref{cor-17}\\
\hline
$q(q+2) $ &  $\frac{q^2+2q+1}{4} $ & $k(k-1)$ & $v-1$ & $(\gf_{q} \times \gf_{q+2},+)$ & Cor. \ref{cor-17}\\
\hline
$q $ &  $\frac{q-3}{4} $ & $\frac{k(k-1)}{2}$ & $\frac{v-1}{2}$ & $(\gf_{q}, +)$ & Cor. \ref{cor-12}\\
\hline
\end{tabular}
\end{table}

\section{The second construction with difference sets}\label{sec-threethree}

Now we introduce the second construction of difference families using difference sets.
The following well-known result is useful in this section.

\begin{lemma}
  \cite{BJL}
  \label{dfstabilizer1}
  Let $G$ be a group of order $v$ and $S=\{B_1,\cdots,B_s\}$ be a family of base blocks whose $G$-stabilizer are
  all trivial. Then $S$ is a $(v,k,\lambda;s)$ difference family if and only if the design $\mathcal{S}$ generated
  by $S$ is a $2$-$(v,k,\lambda)$ design. 
\end{lemma}

We show the following result.

\begin{lemma}
    \label{shouldivede}
    Let $D$ be a $(v,k,\lambda)$ difference set in a group $G$ of order $v$. If there exist $x,g\in G^*\triangleq G\setminus\{1\}$ such that
    $(D_x)^g=(D\cap Dx)^g\triangleq Dg\cap Dxg=D_x$, then $\mbox{ord}(g)\mid \lambda$, where $\mbox{ord}(g)$ denotes the order of $g$ in the group $G$.
\end{lemma}

\begin{proof}
    Assume $x,g\in G^*$ such that $(D_x)^g=D_x$. Let $N=\langle g \rangle$ be the cyclic subgroup generated by $g$.
    Define the group $N$ acting on $D_x$ by $\alpha^{g^i}=\alpha g^i, \forall \alpha\in D_x, g^i\in N$. Clearly this action is faithful.
    Therefore, for each $\alpha\in D_x$, the length of its orbit is $|N|/|N_\alpha|=|N|$, which
    implies that $\mbox{ord}(g)=|N|$ divides the cardinality of $D_x$, which is $\lambda$.
\end{proof}

An immediate useful result is as follows.
\begin{lemma}
    \label{vlambdacoprime}
    Let $D$ be a $(v,k,\lambda)$ difference set in $G$ with order $v$. Assume that $\gcd(v,\lambda)=1$.
    Then the $G$-stabilizer of the block $D_x=D\cap Dx$ is $\{1\}$ for any $x\in G^*$.
\end{lemma}
\begin{proof}
    Since the order of any $g\in G^*$ is a divisor of $v$ and $\gcd(v,\lambda)=1$, we have $\gcd(\mbox{ord}(g), \lambda)=1$.
	By Lemma \ref{shouldivede}, we prove the result.
\end{proof}

We are interested in the difference set $D$ whose $G$-stabilizer\begin{color}{red}{s}\end{color} of all $D_x$ are trivial.
First, if the order $v$ of $G$ is even, then $G$ has involution\begin{color}{red}{s}\end{color}. 
Obviously any involution $g$ may fix the base block $D_g$. 
Indeed, $(D_g)^g=Dg\cap Dg^2=D_g$.
Therefore, we have the following corollary.
\begin{corollary}
Let $G$ be a group of order $v$ and $D$ a $k$-subset of $G$, where $v$ is even.
Assume that $D$ is a $(v,k,\lambda)$-difference set of $G$. Then $\gcd(v,\lambda)>1$.
\end{corollary}

Now a natural question arises: when $\gcd(v,\lambda)>1$, whether there exists any difference set $D$ of $G$
such that all base blocks $D_x=D\cap Dx$ have the trivial $G$-stabilizer under the action $D_x^g=Dg\cap Dxg$?
The following result gives a positive answer. A difference set $D$ is called a \textit{skew Hadamard difference
set} if $D\cap D^{(-1)}=\emptyset, D\cup D^{(-1)}=G\setminus\{1\}$.
A skew Hadamard difference set has always the parameters $(v,(v-1)/2, (v-3)/4)$. 
Skew Hadamard difference sets are related to planar mappings, pseudo-Paley graphs and presemifields, see \cite{weng, DY04}.
They are known to exist in Abelian groups (Corollary \ref{cor-faa}) and non-Abelian groups \cite{feng-skds}.

\begin{lemma}
  \label{skewcase}
  Let $D$ be a $(v,k,\lambda)$ skew Hadamard difference set in $G$ with order $v$. For any $x\in G^*$,
  the $G$-stabilizer of $D_x=D\cap Dx$ is $\{1\}$.
\end{lemma}
\begin{proof}
    Fix an $x\in G^*$, assume that there exists $g\in G^*$ such that $(D_x)^g=Dg\cap Dxg=D_x$. 
	Clearly $|D_x|=\lambda$ and we write $D_x=\{a_1,a_2,\cdots,a_\lambda\}\subset D$.
	Now we have $\{a_1g,a_2g,\cdots,a_\lambda g\}=\{a_1,a_2,\cdots,a_\lambda\}$.
	By Lemma \ref{shouldivede}, we have $\mbox{ord}(g):=t\mid \lambda$.
	Assume that $\lambda=\ell t$ and now we may rewrite $D_x$ as follows (relabel the elements in $D_x$ if necessary):
	\begin{eqnarray*}
		D_x=\{
		    \underbrace{a_1,a_1g,\cdots,a_1g^{t-1}}_{t},\cdots,
	        \underbrace{a_{(\ell-1)t+1},a_{(\ell-1)t+1}g,\cdots,a_{(\ell-1)t+1}g^{t-1}}_t
			\}.
	\end{eqnarray*}
	Now we assume $D=\{a_1,\cdots,a_\lambda,b_1,\cdots,b_{k-\lambda}\}$.
	Notice that $D_x\subset D$ and $g$ can be represented by the form $(a_jg^{i-1})^{-1}\cdot a_jg^{i}$ for $1\le i\le t, 1\le j\le \ell$.
	Therefore, $g$ is represented $\ell t=\lambda$ times as the differences of elements in $D$. 
	This implies that $b_i^{-1}b_j\ne g$ for all $1\le i,j\le k-\lambda$, i.e., $b_jg, b_jg^{-1}\not\in D$ for all $1\le j\le k-\lambda$.
	Since $G^*$ is the disjoint union of $D$ and $D^{ (-1)}$, if all $b_jg, b_jg^{-1}\ne 1$,
	we have $|D^{ (-1)}|=2(k-\lambda)=(v+1)/2\ne (v-1)/2=k$.
	Now w.l.o.g we assume $b_1g=1$ (the case that $b_1g^{-1}=1$ is similar and we omit it), i.e., $g=b_1^{-1}$.
 Since 
	\begin{eqnarray}
		\label{dinverse2}
		D^{(-1)}=\{b_1g^{-1}, b_2g, b_2g^{-1}, \cdots, b_{k-\lambda}g, b_{k-\lambda}g^{-1} \},
	\end{eqnarray}
	we see that $g=b_ig^{-1}$ for some $1 \leq i \leq k-\lambda$ 
	as otherwise $g=b_ig$ for some $1 \leq i \leq k-\lambda$ which 
	would imply that $b_i=1$. Hence $g^2=b_i$ which means that $g^3=b_1^{-1}b_i$ and from this we conclude that $\mbox{ord}(g)=3$.
	On the other hand, $D^{(-1)}$ can be represented alternatively as:
	 \begin{eqnarray}
	  \label{dinverse1}
                     D^{(-1)}=\{
	             a_1^{-1}, g^{-1}a_1^{-1},ga_1^{-1},\cdots,
	             a_{(\ell-1)t+1}^{-1},g^{-1}a_{(\ell-1)t+1}^{-1}, ga_{(\ell-1)t+1}^{-1},
				 g, b_2^{-1},\cdots, b_{k-\lambda}^{-1} 
				 \}.
	\end{eqnarray}
	Observe that $D^{(-1)}$ satisfies the following property: If $x \in D^{(-1)}$ and $gx \in D^{(-1)}$, then $g^2x \in D^{(-1)}$.
	
	First, we show that it is impossible to have $b_rg=b_i^{-1}$ and $b_rg^{-1}=b_j^{-1}$ simultaneously
	for $1 \leq i, j, r \leq k-\lambda$. Suppose this is the case, then we have 
	$b_ig=g^{-1}b_r^{-1}g$, $b_ig^{-1}=g^{-1}b_r^{-1}g^{-1}$, $b_jg=gb_r^{-1}g$ 
	and $b_jg^{-1}=gb_r^{-1}g^{-1}$. Since $gb_r^{-1}g \in D^{(-1)}$ and $g(gb_r^{-1}g) \in D^{(-1)}$, 
	by the observation, we have $g^2(gb_r^{-1}g)=b_r^{-1}g \in D^{(-1)}$, similarly, $b_r^{-1}g^{-1} \in D^{(-1)}$. 
	Suppose $b_r^{-1}g^{-1}=b_ig^{-1}$ for some $1 \leq i \leq \lambda-k$, 
	then $b_r^{-1}=b_i$ which is a contradiction, hence $b_r^{-1}g^{-1}=b_ig$ for
	some $1 \leq i \leq \lambda-k$, and we obtain $b_r^{-1}g=b_i$ but $b_r^{-1}g \in D^{(-1)}$ 
	while $b_i \in D$. Hence, for  a fixed $r$, at least one of $b_rg$ and $b_rg^{-1}$ must take the form of 
	$a_i^{-1}$ for some $1\leq i \leq \lambda$.\\ 
	Next, we consider $b_rg=a_i^{-1}$ and $b_rg^{-1}=b_j^{-1}$. 
	Since $a_i^{-1}=b_rg$, $ga_i^{-1}=gb_rg \in D^{(-1)}$ and $g^{-1}a_i^{-1}=g^{-1}b_rg \in D^{(-1)}$. 
	On the other hand, since $b_j=gb_r^{-1}$ , we have $b_jg=gb_r^{-1}g \in D^{(-1)}$ and $b_jg^{-1}=gb_r^{-1}g^{-1} \in D^{(-1)}$. 
	Suppose $b_jg=a_p^{-1}$ for some $1\leq p \leq \lambda$, then $ga_p^{-1}=g^{-1}b_r^{-1}g \in D^{(-1)}$ 
	but its inverse $g^{-1}b_rg$ is also in  $D^{(-1)}$, so  $b_jg\neq a_p^{-1}$ for all $1\leq p \leq \lambda$. 
	Suppose $b_jg^{-1}=a_p^{-1}$ for some $1\leq p \leq \lambda$, 
	then $ga_p^{-1}=g^{-1}b_r^{-1}g^{-1}$ but its inverse $gb_rg$ is also in $D^{(-1)}$. 
	Similarly, we can rule out the possibility of the case of $b_rg=b_j^{-1}$ and $b_rg^{-1}=a_i^{-1}$ simultaneously.
	Hence for all $i>1$, $b_i^{-1}$ can neither be represented in the form of $b_j^{-1}g$ nor in the form of $b_j^{-1}g^{-1}$ 
	for any $1 \leq j \leq k-\lambda$. Hence, the G-stabilizer of $D_x=D \cap Dx$ is $\{1\}$.
\end{proof}

Now we are ready to describe the second construction of difference families.

\begin{theorem}
   \label{construction2}
   Let $D$ be a $(v,k,\lambda)$ difference set in $G$ with order $v$. Assume $D$ satisfies one of the following cases:
   \begin{equation}
           \label{twokinds}
	   \begin{array}{lll}
	   & (1)& v\ \mbox{is odd and}\ \gcd(v,\lambda)=1, \\
	   & (2)& D\ \mbox{is a skew difference set}.
	 \end{array}
   \end{equation}
   Then for any disjoint union $H_1\cup H_2\cup \{1\}, H_1=H_2^{(-1)}$ of $G$, the two collections 
   $\mathcal{F}_1,\mathcal{F}_2$ of subsets of $G$
   \begin{eqnarray*}
     \mathcal{F}_1=\{ D_x: x\in H_1\}\ \mbox{and}\ \mathcal{F}_2=\{ D_x: x\in H_2\} 
   \end{eqnarray*}
   are both $(v,\lambda,\lambda(\lambda-1)/2)$ difference families of $G$, where $D_x=D\cap Dx$.
\end{theorem}

\begin{proof}
  Let $\mathcal{D}=(G,\mathcal{B})$ be the block design generated by the difference set $D$, where
  $\mathcal{B}=\{ Dx: x\in G\}$. Clearly $\mathcal{D}$ is a symmetric $(v,k,\lambda)$ design.
  Now define the block design $\mathcal{D}^\cap=(G,\mathcal{B}^\cap)$ by 
  $\mathcal{B}^\cap=\{Dx\cap Dy: x, y\in G, x\ne y \}$,
  It is obviously that the collections $H_1,H_2$ are both base blocks for $\mathcal{D}^\cap$.
  It is not difficult to show that $D^\cap$ is a $2-(v,\lambda,\lambda(\lambda-1)/2)$ design.
  By Lemmas \ref{dfstabilizer1}, \ref{vlambdacoprime}, \ref{skewcase}, we see that $\mathcal{F}_1,\mathcal{F}_2$
  are both $(v,\lambda,\lambda(\lambda-1)/2)$ difference families.
\end{proof}

\begin{remark}
  For difference sets which are not of the two types (\ref{twokinds}), Theorem \ref{construction2} may not be true.
  We give an example here.
  Let $D=\{ 1, w, w^2, w^4, w^5, w^8, w^{10} \}$ be a $(15,7,3)$ Singer difference set of $\mathbb{F}_{2^4}^*$,
  where $w$ is a primitive element. Given a disjoint union of $G^*=H_1\cup H_2$,
  where $H_1=\{ w^{11}, w^{13}, w^{14}, w^3, w^5, w^7, w^9 \}, H_2=\{ w, w^{12}, w^2, w^4, w^6, w^8, w^{10} \}, H_1=H_2^{(-1)}.$
  Using MAGMA to compute the ordinary differences of the blocks in $\mathcal{F}_1,\mathcal{F}_2$ we obtain
  $\{2^{12}, 3^2\}$,  which means $12$ elements in $\mathbb{F}_{16}^*$ appears twice and $2$ elements appears three times. 
  Therefore, $\mathcal{F}_1,\mathcal{F}_2$ are not difference families. 
  Notice that here $D$ has the parameters of skew Hadamard difference sets but $D\cap D^{(-1)}\ne \emptyset$. 
\end{remark}

\begin{corollary}\label{cor-faa}
Let $q \equiv 3 \bmod{4}$ be a prime power. If $D$ is the set of all
quadratic residues in $\gf_{q}$, then $\{D_x: x\in D\}$ is a
$(q,(q-3)/4,(q-3)(q-7)/32; (q-1)/2)$ difference family in $(\gf_{q},+)$.
\end{corollary}

\begin{example}
Let $D$ be all the set of quadratic residues modulo $11$. Then $D$ is a
$(11, 5, 2)$ skew difference set in $(\Z_{11}, +)$. The following
blocks form a $(11, 2, 1; 5)$ difference family in $(\Z_{11}, +)$:
$$
\{ 4, 5 \}, \{ 1, 4 \}, \{ 5, 9 \}, \{ 3, 9 \}, \{ 1, 3 \}.
$$
\end{example}

There are several families of skew Hadamard difference sets (see \cite{DY04,DWX, feng-skds}).
They give difference families with similar parameters.

\section{The third construction with difference sets}\label{sec-fourfour}

We will give the third construction of different families in this section.
In order to obtain more general results, we will first use balanced incomplete block designs
for the construction. As an application, difference families are obtained using difference sets.

A $(v, k, \lambda, r, b)$ incidence structure $\mathcal{D}=(\mathcal{V},\mathcal{B})$ is called 
a \textit{balanced incomplete block design} (BIBD) if the vertex set $\mathcal{V}$ with $v$ points
is partitioned into a family $\mathcal{B}$ of $b$ subsets (blocks) in such a way that any two points 
determine $\lambda$ blocks with $k$ points in each block, 
and each point is contained in $r$ different blocks. 
For a BIBD $\mathcal{D}$, it is well known that $r,b$ can be determined by $v,k,\lambda$ as follows:
\begin{equation*}
  b=\frac{v(v-1)\lambda}{k(k-1)}\ \mbox{and}\ r=\frac{\lambda(v-1)}{k-1}.
\end{equation*}
Therefore, a $(v,k,\lambda, r,b)$-BIBD is commonly written as a $(v,k,\lambda)$-BIBD.
In the following, for a set $X$ and an integer $0\le s\le |X|$, the notation $\left(X \atop s\right)$
means the collection of the subsets of $X$ with cardinality $s$.

\begin{theorem}
  \label{BIBDconstruction}
  Let $\mathcal{D}=(\mathcal{V},\mathcal{B})$ be a $(v,b,r,k,\lambda)$-BIBD. Define
 \begin{eqnarray*}
   && \mathcal{B}_s^+=\{ B+X| B\in \mathcal{B}, X\in\left(V-B \atop{s}\right)\}\ \mbox{for}\ 0\le s\le v-k\ \mbox{and} \\
   && \mathcal{B}_s^-=\{ B-X| B\in\mathcal{B}, X\in \left(B \atop{s}\right)\}\ \mbox{for}\ 0\le s\le k.
 \end{eqnarray*}
 Then $(V,\mathcal{B}_s^+)$ is a $(v,k+s,\lambda_s^+)$-BIBD and $(V,\mathcal{B}_s^-)$ is a $(v,k-s,\lambda_s^-)$-BIBD, where
 \begin{eqnarray*}
   && \lambda_s^+=\lambda\left(v-k \atop{s}  \right)+2(r-\lambda)\left(v-k-1 \atop{s-1}\right)+(b-2r+\lambda)\left( v-k-2 \atop{s-2} \right)\ \mbox{and}\\ 
   && \lambda_s^-=\lambda\left(k-2 \atop{s}\right).
 \end{eqnarray*}
\end{theorem}
\begin{proof}
 First, we prove $\lambda_s^+$ is a $(v,k+s,\lambda_s^+)$-BIBD. Clearly we only need to prove for any two points 
 $x,y\in\mathcal{V}$, there are $\lambda_s^+$ blocks containing them.
 Now consider the following three cases of the block $B\in\mathcal{D}$:
 
Case $(1)$: the block $B$ contains both $x$ and $y$. Clearly for each such $B$, there are 
 $\left(v-k \atop{s}  \right)$ choices of $X$ such that $x,y\in B+X$. Since $\mathcal{D}$
 is a $(v,k,\lambda)$-BIBD, we have $\lambda\left(v-k \atop{s}  \right)$ blocks $B+X$ in this case.
 
Case $(2)$: the block $B$ contains one and only one of $x$ and $y$. Now assume $x\in B, y\not\in B$. 
Clearly there are $r-\lambda$ such blocks $B$.
 For each of such blocks $B$, we choose a subset $X$ of $V-B$ containing $y$, which have $\left( v-k-1, s-1\right)$
 choices. Therefore, we have altogether $(r-\lambda)\left( v-k-1 \atop{s-1}\right)$ blocks $B+X$. A calculation of the
 symmetric case $x\not\in B, y\in B$ contributes another $(r-\lambda)\left( v-k-1 \atop{s-1}\right)$ blocks.

Case $(3)$: the block $B$ does not contain any of $x$ and $y$. Now assume $x,y\not\in B$, then clearly there are $b-\lambda-2(r-\lambda)=b-2r+\lambda$
 such blocks $B$. For each such block $B$, there are $\left( v-k-2 \atop{s-2}\right)$ choices $X$ such that $x,y\in B+X$.
 Therefore, in this case there are altogether $(b-2r+\lambda)\left(v-k-2 \atop{s-2} \right)$ blocks $B+X$ such that $x,y$ are both in. 
 Combining cases $(1), (2) ,(3)$, we prove the result.
 Similar arguments may show $\mathcal{B}_s^-$ is a $(v,k-s,\lambda^-)$-BIBD, we omit it here.
\end{proof}

Now let $D$ be a $(v,k,\lambda)$ difference set of a group $G$. 
Then the block design $\mathcal{D}$ generated by $D$ is a $(v,k,\lambda)$-BIBD.
Apply Theorem \ref{BIBDconstruction} on $\mathcal{D}$, we have 
the following two constructions of difference families using
difference sets. Notice that $b=v, r=k$ for $\mathcal{D}$.

\begin{corollary}
  \label{construction3}
Let $G$ be a group of order $v$ and $D$  a $(v, k, \lambda)$ difference set in $G$.
Let $s$ be an integer with $1\le s\le v-k-1$.
For each $X \in \left(G-D \atop{s} \right)$, define
$$
D_X = D+X.
$$
Then $\{D_X: X \in \left(G-D \atop{s} \right) \}$ is a
$(v, k+s, \lambda_s^+; u)$ difference family in $G$,
where
$$ \lambda_s^+=\lambda\left(v-k \atop{s}  \right)+2(k-\lambda)\left(v-k-1 \atop{s-1}\right)+(v-2k+\lambda)\left( v-k-2 \atop{s-2} \right), 
   u=\left(v-k \atop{s} \right).  $$
\end{corollary}

\begin{corollary}
    \label{construction5}
    Let $G$ be a group of order $v$ and $D$ a $(v, k, \lambda)$ difference set.
    Let $s$ be an integer with $1\le s\le k-1$.
    For each $X\in\left( D \atop{s}\right)$, define
    $$
    D_X = D - X.
    $$
    Then $\{D_X: X \in \left( D \atop{s} \right) \}$ is a $(v, k-s, \lambda_s^-; u)$ difference family in $G$,
    where $$\lambda_s^-=\lambda\left(k-2 \atop{s-2} \right), u=\left(k \atop{s}\right).$$
\end{corollary}

\section{The forth construction with difference sets}\label{sec-fivefive}

Now we present the final construction of difference families using difference sets.

\subsection{The construction}\label{sec-dsDF}
\begin{theorem}
  \label{thm-mainmain}
  Let $G$ be an Abelian group and $D$ a $(v,k,\lambda)$ difference set of $G$.
  Let $H$ be a subgroup of $G$ and $|G:H|=\ell$. Let $1=g_0,g_1,\ldots,g_{\ell-1}$
  be a complete set of coset representatives of $G$.
  For each $g_i$, define $D_ig_i=D\cap Hg_i$.
  Then $\{D_0,D_1,\ldots,D_{\ell-1}\}$ is a $(v,K,\lambda)$ difference family of $H$,
  where $K=\{|D_i|: 0\le i\le \ell-1\}$.
\end{theorem}
\begin{proof}
  First notice that $D_i\subset H$.
  Now it follows from $D=D\cap G=\bigcup\limits_{i=0}^{\ell-1}(D\cap Hg_i)$ and $D$ is a $(v,k,\lambda)$ difference set that
  \begin{eqnarray*}
    DD^{(-1)}  &=& \left(\sum\limits_{i=0}^{\ell-1} D_ig_i\right)\left(\sum\limits_{j=0}^{\ell-1} D_i^{(-1)}g_i^{-1}\right) \\
    &=& \sum\limits_{i,j=0}^{\ell-1} D_iD_j^{(-1)} g_ig_j^{-1} \\
    &=& k-\lambda+\lambda\left(Hg_0+Hg_1+\cdots +Hg_{\ell-1} \right).
  \end{eqnarray*}
  Comparing the terms which fall into $H$ we have
  $$
  \sum\limits_{i=0}^{\ell-1} D_iD_i^{(-1)}=k-\lambda+\lambda H,
  $$
  which implies that $\{D_0,D_1,\cdots,D_{\ell-1}\}$ is a $(v,K,\lambda)$ difference family of $H$.
\end{proof}

This is a generic construction in the sense that it works for any Abelian difference set. 
The numbers $k_i=|D_i|$ for $0\le i\le \ell-1$ are called \textit{intersection numbers with respect to $H$} (see \cite[P. 331-332]{BJL}).
However, the cardinalities $k_i$ of the base blocks $D_i$
may be hard to determine in some cases. This will be seen in subsequent
sections when we deal with difference families obtained from specific
classes of cyclic difference sets.

\subsection{Determination of intersection numbers}

\subsubsection{General bounds}
We first give some basic properties of the intersection numbers $k_i$'s as following. A proof may be found in \cite[Lemma 5.4]{BJL}.
\begin{lemma}
  \label{thm-relations}
  Let the symbols and notations be the same as above. For the cardinalities $k_i$
  we have the following.
  \begin{enumerate}
     \item $\sum_{i=0}^{\ell-1} k_i = k$.
     \item $\sum_{i=0}^{\ell-1} k_i^2 = \lambda (n-1)+k$.
     \item $\sum_{i=0}^{\ell-1} k_i k_{(i+\tau_2) \bmod \ell} = \lambda n$ for
         each $\tau_2$ with $0 < \tau_2 < \ell$.
  \end{enumerate}
\end{lemma}

Note that the three relations among the cardinalities $k_i$ described in
Lemma \ref{thm-relations} are similar to those for perfect nonlinear
functions described in \cite{CD2004}. In addition, these $k_i$ are the
same as those developed by Baumert \cite{Baum}.

In view that it is very hard to determine these parameters $k_i$ in some cases, the following theorem is useful. The bounds on $k_i$ given in Theorem \ref{bound} are general ones. For specific difference sets $D$ and groups $G$, better bounds may exist. For instance in \cite[Lemma 2.2]{momihara}, the author gave the bounds for $k_i$ when $D$ is a cyclic difference set in $(\gf_{p^m},+)$.

\begin{theorem}
	\label{bound}
Let the symbols and notations be the same as above. Then
\begin{eqnarray*}
	\frac{k}{\ell}-\sqrt{\frac{(k-\lambda)(\ell-1)}{\ell}}\le k_i\le \frac{k}{\ell}+\sqrt{\frac{(k-\lambda)(\ell-1)}{\ell}},\ 0\le i\le \ell-1.
\end{eqnarray*}
\end{theorem}
\begin{proof}
	Let $k_i=\frac{k}{\ell}+\lambda_i$ for $0\le i\le \ell-1$. It is clear that $\sum\limits_{i=0}^{\ell-1}\lambda_i=0$
	as $\sum\limits_{i=0}^{\ell-1}k_i=k$. By Lemma \ref{thm-relations},
    \begin{eqnarray*}
		\begin{array}{lll}
		\lambda(n-1)+k &=& \sum\limits_{i=0}^{\ell-1}k_i^2 = \sum\limits_{i=0}^{\ell-1}(\frac{k^2}{\ell^2}+\lambda_i^2+\frac{2k\lambda_i}{\ell}) \\
					  &=& \frac{k^2}{\ell}+\sum\limits_{i=0}^{\ell-1}\lambda_i^2. \\
		\end{array}
    \end{eqnarray*}
	Therefore,
	\begin{eqnarray*}
		\begin{array}{lll}
           \sum\limits_{i=0}^{\ell-1}\lambda_i^2 &=& \lambda(n-1)+k-\frac{k^2}{\ell}=\frac{1}{\ell}(\lambda(n\ell-\ell)+k\ell-k^2) \\ [2ex]
		   &=&\frac{1}{\ell}(\lambda(v-1)-\lambda(\ell-1)+k\ell-k^2)=\frac{1}{\ell}(k(k-1)-\lambda(\ell-1)+k\ell-k^2)\\ [2ex]
		   &=&\frac{1}{\ell}(k-\lambda)(\ell-1). \\
		\end{array}
	\end{eqnarray*}
	Now we have
	$$ -\sqrt{\frac{(k-\lambda)(\ell-1)}{\ell}}\le \lambda_i\le \sqrt{\frac{(k-\lambda)(\ell-1)}{\ell}} $$
	and the proof is completed.
\end{proof}

\subsubsection{The case $\ell=2,3,4$}
Next we determine the intersection numbers $k_i$'s for $\ell=2,3,4$.
\begin{theorem}
	\label{lis2}
	If $\ell=2$, then $(k_0,k_1)=(\frac{k+\sqrt{k-\lambda}}{2}, \frac{k-\sqrt{k-\lambda}}{2})\ \mbox{or}\
	(\frac{k-\sqrt{k-\lambda}}{2}, \frac{k+\sqrt{k-\lambda}}{2})$.
\end{theorem}
\begin{proof}
	The proof is straightforward by replacing $\ell=2$ in Lemma \ref{thm-relations}.
\end{proof}

An immediate corollary on the existence of Abelian difference sets may be derived.
\begin{corollary}
   Let $D$ be a $(v,k,\lambda)$ difference set in $G$ with order $v$, where $v$ is an even integer.
   Then $k-\lambda$ is a square.
\end{corollary}

As a corollary of Theorem \ref{lis2}, we have the following.
\begin{corollary}
If $D$ is a cyclic difference set with parameters $(4u^2, 2u^2-u, u^2-u)$ and $\ell =2$,
then we have
$$
(k_0,k_1)=(u^2, u^2-u) \mbox{ or } (k_0,k_1)=(u^2-u, u^2).
$$
\end{corollary}

The determination of the parameters $k_i$ for the case $\ell =3$ is nontrivial. The following lemma will be used for the case $\ell=3$.

\begin{lemma}
	\cite{nair}
	\label{nairs}
	Let $a$ be a positive integer. Then the equation $x^2+y^2+xy=a$ is solvable (has integer solutions) if and only if, for every
	prime $p$ dividing $a$ such that $p\ne 3$ and $p\not\equiv 1\bmod 6$, an even power of $p$ exactly divides $a$.
\end{lemma}

Let $\ell=3$. To determine $k_0,k_1,k_2$, we need to solve the following system of equations:
\begin{eqnarray}
	\label{eqlis3}
\left\{
    \begin{array}{ll}
	 & k_0+k_1+k_2=k, \\
	 & k_0^2+k_1^2+k_2^2=\lambda(n-1)+k, \\
	 & k_0k_1+k_1k_2+k_0k_2=\lambda n. \\
	\end{array}
\right.
\end{eqnarray}

Two solutions $(k_0,k_1,k_2)$ and $(k_0',k_1',k_2')$ of (\ref{eqlis3}) are  \textit{equivalent} if
$\{k_0, k_1, k_2\}=\{k_0', k_1', k_2'\}$ as multisets.
It follows from the second and the third equation of (\ref{eqlis3}) that
$$(k_0+k_1)^2+(k_1+k_2)^2+(k_0+k_2)^2=4\lambda n+2(k-\lambda).$$
Let $x=k_0+k_1=k-k_2, y=k_1+k_2=k-k_0$. We have then
\begin{eqnarray}
  \label{step1}
  x^2+y^2+(2k-x-y)^2=4\lambda n+2(k-\lambda).
\end{eqnarray}
Simplifying (\ref{step1}) gives that
\begin{eqnarray}
	\label{step2}
    x^2+y^2+xy-2k(x+y)=-\lambda n-k^2.
\end{eqnarray}

In the following we distinguish between the cases $k\equiv 0\bmod 3$ and $k\not\equiv 0\bmod 3$.

{\bf{Case 1}}. Assume that $k\equiv 0\bmod 3$. Setting $s=x-\frac{2k}{3}, t=y-\frac{2k}{3}$,
we obtain from (\ref{step2}) that
\begin{equation}
	\label{step13}
	s^2+t^2+ts=\frac{1}{3}(k-\lambda).
\end{equation}
The above computations use the fact that $k(k-1)=\lambda(v-1)$.
The following result is obvious from the above discussions.

\begin{lemma}
	\label{case1change}
	If $k\equiv 0\bmod 3$, then the system of equations (\ref{eqlis3}) is solvable if and only if
	(\ref{step13}) is solvable.
\end{lemma}

By Lemmas \ref{case1change} and \ref{nairs}, we have the following.

\begin{theorem}
	\label{case1}
	Let the symbols and notations be as above. If $k\equiv 0\bmod 3$, then the system of equations (\ref{eqlis3}) is solvable if and only if the canonical factorization of the integer
	$\frac{1}{3}(k-\lambda)$ has the following form
	$$\frac{1}{3}(k-\lambda)=a^2\cdot 3^b\cdot p_1^{\alpha_1}\cdots p_t^{\alpha_t},$$
    where no prime divisors of $a$ are of the form $6k+1$, and $b,t,\alpha_i$ are integers, and
	$p_i$ are distinct primes with $p_i\equiv 1\bmod 6$ for $1\le i\le t$.
	Moreover, when the set of equations (\ref{eqlis3}) is solvable, all the inequivalent solutions of (\ref{eqlis3})
	are:
	$$\left(\frac{k}{3}-t', \frac{k}{3}+s'+t', \frac{k}{3}-s'\right),\ \left(\frac{k}{3}+t', \frac{k}{3}-t'-s', \frac{k}{3}+s'\right), $$
	where $s'^2+s't'+t'^2=\frac{1}{3}(k-\lambda)$ and $s'\ge t'$.
\end{theorem}

{\bf{Case 2}}. If $k\not\equiv 0\bmod 3$,  multiply $9$ on both sides of (\ref{step2}) and set
$s=3x-2k, t=3y-2k$. By (\ref{step2}) we have
\begin{eqnarray*}
    (s+2k)^2+(t+2k)^2+(s+2k)(t+2k)-6k(s+t+4k)=-9(\lambda n+k^2).
\end{eqnarray*}
Simplifying the above equation yields that
\begin{eqnarray}
	\label{step23}
    s^2+t^2+st=3(k-\lambda).
\end{eqnarray}

For an arbitrary solution $(s,t)$ of (\ref{step23}), if (\ref{eqlis3}) has a solution $(k_0,k_1,k_2)$,
then $(s+2k\bmod 3, t+2k\bmod 3)=(0,0)$. This is equivalent to saying that $(s\bmod 3, t\bmod 3)=(1,1)$ if $k\equiv 1\bmod 3$;
and $(s\bmod 3, t\bmod 3)=(2,2)$ if $k\equiv 2\bmod 3$. We have the following result.

\begin{theorem}
    \label{case2}
	Let the symbols and notations be as above. If $k\not\equiv 0\bmod 3$, then the system of equations (\ref{eqlis3}) is solvable if
	$3(k-\lambda)$ is a positive integer with the following canonical factorization
	$$3(k-\lambda)=a^2\cdot 3^b\cdot p_1^{\alpha_1}\cdots p_t^{\alpha_t},$$
    where no prime divisors of $a$ are of the form $6k+1$, and $b,t,\alpha_i$ are integers with $b\ge 1$, and
	$p_i$ are distinct primes with $p_i\equiv 1\bmod 6$ for $1\le i\le t$.
	Moreover, assume that the system of equations (\ref{eqlis3}) are solvable. Let $(s',t')$ be a solution with
	$(s'\bmod 3, t'\bmod 3)=(1,1)\ (\mbox{resp.}\ (2,2) )$ when $k\equiv 1\ (\mbox{resp}\ 2)\bmod 3$, then
	all the inequivalent solutions of (\ref{eqlis3})
	are:
	$$\left(\frac{k+t'}{3}, \frac{k+s'+t'}{3}, \frac{k-s')}{3}\right),\ \left(\frac{k-t'}{3}, \frac{k-s'-t'}{3}, \frac{k+s'}{3}\right),$$
	where $s'^2+s't'+t'^2=3(k-\lambda)$ and $s'\ge t'$.
\end{theorem}

\begin{remark}
Usually it is difficult to find out the solutions of the equation $x^2+xy+y^2=a$ for arbitrary $a$ and
there does not exist an unifying method (see \cite{nair}). However, for specific $a$, it is possible to determine the solutions.
For example, when $a=3^r$, the solutions can be found in \cite{lll}.
\end{remark}

We now deal with the case that $\ell =4$, and will need the following lemma proved by Fermat in the sequel \cite{Niven}.

\begin{lemma}\label{lem-Fermat}
Write the canonical factorization of $N$ in the form
$$
N=2^\alpha \prod_{p \equiv 1 \bmod{4}} p^{\beta_p}  \prod_{q \equiv 3 \bmod{4}} q^{\gamma_q}.
$$
Then $N$ can be expressed as a sum of two squares of integers if and only if all the exponents $\gamma_q$
are even.

Also, $N$ is a sum of two relatively prime squares if and only if it is not divisible by 4 and not divisible by any prime congruent to 3, modulo 4.

\end{lemma}

We need also the following.

\begin{theorem} \label{thm-april1}
If a prime $p \equiv 3 \bmod{4}$ divides $a^2+b^2$ for two integers $a$ and $b$, then $p$ divides both $a$ and $b$.
\end{theorem}

Let $v=4u^2$, and let $D$ be a cyclic Hadamard difference set in $\Z_{4u^2}$ with parameters $(4u^2, 2u^2-u^2, u^2-u; u^2)$.
It is known that $u$ must be odd.
As before, we use $k_i$ to denote the cardinalities of the base blocks $D_i$ obtained from $D$ when $n=u^2$ and $\ell=4$.
In this case, the relations among these $k_i$ in Theorem \ref{thm-relations} become
\begin{eqnarray}\label{eqn-4lrela}
\left\{
 \begin{array}{l}
k_0+k_1+k_2+k_3=2u^2-u, \\
k_0^2+k_1^2+k_2^2+k_3^2=u^4-u^3+u^2, \\
k_0k_2+k_1k_3=\frac{(u^2-u)u^2}{2}.
\end{array}
\right.
\end{eqnarray}
Combining the second and the third equations in (\ref{eqn-4lrela}) yields
\begin{eqnarray}\label{eqn-babaya}
u^2=(k_0-k_2)^2+(k_1-k_3)^2.
\end{eqnarray}

By Theorem \ref{thm-april1}, Lemma \ref{lem-Fermat} and (\ref{eqn-babaya}),  let $u=u_1(r^2+s^2)$ for an odd integer $u_1$,
two integers $r$ and $s$ with $r>s$ and $\gcd(r, s)=1$, where $u_1$ contains only prime factors
congruent to 1 modulo 4, and the factorization of $r^2+s^2$ has only prime factors congruent to 3 modulo 4.
Then it is well known that all the solutions of (\ref{eqn-babaya}) are of the form
\begin{eqnarray}\label{eqn-april11}
\left\{
\begin{array}{l}
k_0-k_2=u_1(r^2-s^2), \\
k_1-k_3=2u_1rs, \\
u=u_1(r^2+s^2).
\end{array}
\right.
\end{eqnarray}

Combining (\ref{eqn-4lrela}) and (\ref{eqn-april11} ), we have proved the following.

\begin{theorem}
Let $v=4u^2$, and let $D$ be a cyclic Hadamard difference set in $\Z_{4u^2}$ with parameters $(4u^2, 2u^2-u^2, u^2-u, u^2)$.
As before, we use $k_i$ to denote the cardinalities of the base blocks $D_i$ obtained from $D$ when $n=u^2$ and $\ell=4$.
Then we have
$$
k_0=\frac{u^2+u - 2u_1s^2}{2},  k_2=\frac{u^2-u + 2u_1s^2}{2},
k_1=\frac{u^2-u + 2u_1rs}{2},  k_3=\frac{u^2-u - 2u_1rs}{2}.
$$
\end{theorem}

For any given $u$, $u_1$ is fixed, but the factorization of $u/u_1$ into $r^s+s^2$ may not be unique. The system of equations in
(\ref{eqn-4lrela}) may have more than one solutions $(k_0, k_1, k_2, k_3)$.

\subsubsection{Specific constructions with Singer difference sets}

In this section, we will determine the intersection numbers in the case that $D$ is the
Singer difference sets in $(\gf_{q^m},+)$. Fist, we present results on group characters, cyclotomy
and Gaussian sums which will be needed in the sequel.

Let $\tr_{q/p}$ denote the trace function from $\gf_q$ to $\gf_p$.
An {\em additive character} of $\gf_q$ is a nonzero function $\chi$
from $\gf_q$ to the set of complex numbers such that
$\chi(x+y)=\chi(x) \chi(y)$ for any pair $(x, y) \in \gf_q^2$.
For each $b\in \gf_q$, the function
\begin{eqnarray}\label{dfn-add}
\chi_b(c)=e^{2\pi \sqrt{-1} \tr_{q/p}(bc)/p} \ \ \mbox{ for all }
c\in\gf_q
\end{eqnarray}
defines an additive character of $\gf_q$. When $b=0$,
$\chi_0(c)=1 \mbox{ for all } c\in\gf_q,
$
and is called the {\em trivial additive character} of
$\gf_q$. The character $\chi_1$ in (\ref{dfn-add}) is called the
{\em canonical additive character} of $\gf_q$.

A {\em multiplicative character} of $\gf_q$ is a nonzero function
$\psi$ from $\gf_q^*$ to the set of complex numbers such that
$\psi(xy)=\psi(x)\psi(y)$ for all pairs $(x, y) \in \gf_q^*
\times \gf_q^*$.
Let $g$ be a fixed primitive element of $\gf_q$. For each
$j=0,1,\ldots,q-2$, the function $\psi_j$ with
\begin{eqnarray}\label{dfn-mul}
\psi_j(g^k)=e^{2\pi \sqrt{-1} jk/(q-1)} \ \ \mbox{for } k=0,1,\ldots,q-2
\end{eqnarray}
defines a multiplicative character of $\gf_q$. When $j=0$,
$ \psi_0(c)=1 \mbox{ for all }
c\in\gf_q^*,
$
and is called the {\em trivial multiplicative
character} of $\gf_q$.

Let $q$ be odd and $j=(q-1)/2$ in (\ref{dfn-mul}), we then get a
multiplicative character $\eta$ such that $\eta(c)=1$ if $c$ is
the square of an element and $\eta(c)=-1$ otherwise. This $\eta$
is called the {\em quadratic character} of $\gf_q$.

Let $r$ be a prime power and $r-1=nN$ for two positive integers $n>1$ and $N>1$, and let
$\alpha$ be a fixed primitive element of $\gf_r$.
Define $C_{i}^{(N,r)}=\alpha^i \langle \alpha^{N} \rangle$ for $i=0,1,...,N-1$, where
$\langle \alpha^{N} \rangle$ denotes the
subgroup of $\gf_r^*$ generated by $\alpha^{N}$. The cosets $C_{i}^{(N,r)}$ are
called the {\em cyclotomic classes} of order $N$ in $\gf_r$.

The \textit{cyclotomic number} of order $N$, denoted $(i,j)_N$, are defined as
$$
(i,j)_N=|(C_i^{(N,r)}+1)\bigcap C_j^{(N,r)}|,
$$
where $0\le i, j\le N-1$ and $|A|$ denotes the number of elements in the set $A$.

The {\em Gaussian periods} are defined by
$$
\eta_i^{(N,r)} =\sum_{x \in C_i^{(N,r)}} \chi(x), \quad i=0,1,..., N-1,
$$
where $\chi$ is the canonical additive character of $\gf_r$.

The values of the Gaussian periods are in general very hard to compute.
However, they can be computed in a few cases.
The following is the classical result of uniform cyclotomy.
\begin{lemma}\label{lem-semipri}
	\cite{baumert}
	Assume that $p$ is a prime, $e\ge 2$ is a positive integer, $r=p^{2j\gamma}$,
	where $N|(p^j+1)$ and $j$ is the smallest such positive integer. Then the
	cyclotomic periods are given by

	{\bf{Case A.}} If $\gamma, p, \frac{p^j+1}{N}$ are all odd, then
	$$\eta_{N/2}=\frac{(N-1)\sqrt{r}-1}{N},\ \eta_{i}=-\frac{1+\sqrt{r}}{N},\ \mbox{for all}\ i\ne\frac{N}{2}.$$

    {\bf{Case B.}} In all the other cases,
	$$\eta_0=-(-1)^\gamma\sqrt{r}+\frac{(-1)^\gamma\sqrt{r}-1}N,\ \eta_i=\frac{(-1)^\gamma\sqrt{r}-1}{N},\ \mbox{for all}\ i\ne 0.$$
\end{lemma}

In this section we consider the cyclic difference families obtained from the Singer difference sets with
parameters $\left(\frac{q^m-1}{q-1}, \frac{q^{m-1}-1}{q-1}, \frac{q^{m-2}-1}{q-1}\right)$. Our
task is to determine the cardinalities $k_i$ of the base blocks $D_i$.

First we give a construction of cyclic difference sets with Singer parameters.
Let $q$ be a power of a prime $p$, and let $m$ be a positive integer. Let $g$ be a generator of $\fqm$.
Define $\alpha=g^{q-1}$ and
$$
D_b^{(m,q)}= \left\{0 \le i < \frac{q^m-1}{q-1}: \Tr(\alpha^i)=b \right\}, \ b \in \fq,
$$
where $\Tr(x)$ is the trace function from $\fqm$ to $\fq$.

We will need the following lemma.

\begin{lemma}
	\label{pre}
    Let notations be as above. Then
	
	$(1)$ $\gcd(q-1,m)=\gcd(q-1,\frac{q^m-1}{q-1})$;

	$(2)$ $\fqm^*=H\times\fq^*$ if and only if $\gcd(q-1,m)=1$, where
	$H=\langle \alpha \rangle$ is the cyclic group generated by $\alpha$ in the multiplicative group of $\fqm$;

	$(3)$ $|D_0^{(m,q)}|=\frac{q^{m-1}-1}{q-1}$ if and only if $\gcd(q-1,m)=1$.
\end{lemma}
\begin{proof}
	(1) Obviously.

	(2) Note that $H=\langle \alpha \rangle=\langle g^{q-1} \rangle$ and $\fq^*=\langle g^{\frac{q^m-1}{q-1}}\rangle$.
	We only need to prove that $H\cap\fq^*=\{1\}$ if and only if $\gcd(q-1,m)=1$. This follows from
	$\gcd(q-1, \frac{q^m-1}{q-1})=\gcd(q-1, m)$.

	(3) Let $G=\Z_{(q^m-1)/(q-1)}$.
	Let $\xi_p$ be a primitive $p$-th root of unity and $\chi(x)=\xi_p^{\Tr_{q^m/p}(x)}$, where $x\in\fqm$.
    It is clear that $\chi$ is an additive character of $\fqm$. By (2) in this Lemma, we have
	$\sum\limits_{y\in\fqstar}\chi(yH)=\chi(\fqmstar)=-1$ if and only if $\gcd(q-1,m)=1$.
	Now
    \begin{eqnarray}
	\begin{array}{lll}
    |D_0^{(m,q)}|&=&\frac{1}q\sum\limits_{y\in\fq}\sum\limits_{x\in G}\xi_p^{\Tr_{q/p}(y\Tr_{q^m/q}(\alpha^{x}))} \\[2ex]
	         &=&\frac{1}q\sum\limits_{y\in\fq}\sum\limits_{x\in G}\chi(y\alpha^x) \\ [2ex]
			 &=&\frac{1}q(|G|+\sum\limits_{y\in\fqstar}\sum\limits_{x\in G}\chi(y\alpha^{x})) \\ [2ex]
			 &=&\frac{1}q\left(\frac{q^m-1}{q-1}+\sum\limits_{y\in\fqstar}\chi(yH)\right) \\ [2ex]
			 &=&\frac{1}q\left(\frac{q^m-1}{q-1}+\chi(\fqm^*)\right) \\ [2ex]
			 &=&\frac{1}q\left(\frac{q^m-1}{q-1}-1\right) \\ [2ex]
			 &=&\frac{q^{m-1}-1}{q-1}.
	\end{array}
    \end{eqnarray}
    The proof is completed.
\end{proof}

Now we describe a family of difference sets with Singer parameters.

\begin{theorem}\label{thm-new}
Let notations be as above. Then
$D_0^{(m,q)}$ is a $$\left(\frac{q^m-1}{q-1}, \frac{q^{m-1}-1}{q-1},\frac{q^{m-2}-1}{q-1}\right)$$
difference set in $\Zz_{(q^m-1)/(q-1)}$ if and only if $\gcd(q-1, m)=1$.
\end{theorem}
\begin{proof}
	It is equivalent to proving that $D=\{x\in H| \tr(x)=0 \}$ is a
    $\left(\frac{q^m-1}{q-1}, \frac{q^{m-1}-1}{q-1},\frac{q^{m-2}}{q-1}\right)$
    difference set in $H$ if and only if $\gcd(q-1, m)=1$,
	where $H=\langle \alpha \rangle$ is the cyclic group generated by $\alpha$ in the multiplicative group of $\fqm$.
	By Lemma \ref{pre} (3), $|D|=\frac{q^{m-1}-1}{q-1}$ if and only if $\gcd(q-1,m)=1$.
    Now define a mapping $T_\alpha: H\rightarrow H$ by $T_\alpha(x)=x\alpha$.
	It is easy to verify that $T_\alpha$ is an automorphism of the development of $D$, say $\mathcal{D}=Dev(D)$.
	Clearly $\langle T_\alpha \rangle\cong H$ and $\langle T_\alpha \rangle$ acts sharply transitive on $\mathcal{D}$,
	then by \cite[Theorem 1.6]{BJL}, $D$ is a $\left(\frac{q^m-1}{q-1}, \frac{q^{m-1}-1}{q-1},\frac{q^{m-2}-1}{q-1}\right)$
       difference set.
	We complete the proof.
\end{proof}

The $D_0^{(m,q)}$ of Theorem \ref{thm-new} is a difference set only under the condition that  $\gcd(q-1, m)=1$ and
should be equivalent to the classical Signer difference set. However, we do not have a proof of the equivalence.

We assume $\gcd(q-1,m)=1$ from now on in this section.
Let $\frac{q^m-1}{q-1}=n\ell$ with $n>1$ and $\ell>1$. Define
$$
\Omega_i=\left\{\left(\frac{q^m-1}{q-1}j+(q-1)i\right) \bmod (q-1)\ell: 0\le j\le q-2 \right\}.
$$
Note that
$$
\gcd((q^m-1)/(q-1), q-1)=\gcd(q-1,m)=1.
$$
We have
$$
\gcd(\ell, q-1)=\gcd(n, q-1)=1.
$$
It can the be seen that
\begin{eqnarray}\label{eqn-bbb}
\Omega_i = \{(\ell j + i(q-1)) \bmod \ell (q-1): 0 \le j \le q-2\}.
\end{eqnarray}

Let $D=D_0^{(m,q)}$ in Theorem \ref{thm-new} and
for $0\le i\le n-1$ define
$$D_i=\{x: 0\le x \le n-1 | \Tr(\alpha^{\ell x+i})=0\}.$$
Let $\xi_p$ be a primitive $p$-th root of unity and $\chi(x)=\xi_p^{\Tr_{q^m/p}(x)}$, where $x\in\fqm$.
It is clear that $\chi$ is an additive character of $\fqm$. It then follows from (\ref{eqn-bbb}) that
\begin{eqnarray}\label{eqn-ki}
	\begin{array}{lll}
		k_i=|D_i|&=&\frac{1}q\sum\limits_{y\in\fq}\sum\limits_{x=0}^{n-1}\xi_p^{\Tr_{q/p}(y\Tr_{q^m/q}(\alpha^{\ell x+i}))} \\
	         &=&\frac{1}q\sum\limits_{y\in\fq}\sum\limits_{x=0}^{n-1}\chi(y\alpha^{\ell x+i}) \\
			 &=&\frac{1}q(n+\sum\limits_{y\in\fqstar}\sum\limits_{x=0}^{n-1}\chi(y\alpha^{\ell x+i})) \\
			 &=&\frac{1}q(n+\sum\limits_{y\in\fqstar}\sum\limits_{x=0}^{n-1}\chi(yg^{(q-1)i}\cdot g^{(q-1)\ell x})) \\
			 &=&\frac{1}q\left(n+\sum\limits_{y\in\fqstar}\sum\limits_{t\in C_{0}^{((q-1)\ell, q^m)}}\chi(yg^{(q-1)i}\cdot t)\right) \\
			 &=&\frac{1}q\left(n+\sum\limits_{j\in\Omega_i}\chi(C_{j}^{((q-1)\ell, q^m)})\right) \\
			 &=&\frac{1}q\left(n+ \chi(C_{i'}^{(\ell, q^m)})\right),
	\end{array}
\end{eqnarray}
where $i'=i(q-1) \bmod{\ell}$.

Now we are ready to determine the cardinalities $k_i$ in certain special cases.
First we have the following result.

\begin{lemma}
	\label{conditionofl}
    $(1)$ $\gcd(\ell, 2)$=1.

    $(2)$ $\sum\limits_{i=0}^{\ell-1}\chi(C_{i'}^{(\ell, q^m)})=-1$, where $i'=i(q-1) \bmod{\ell}$.
\end{lemma}
\begin{proof}
	(1) If $q$ is odd, then clearly $m$ is odd as $\gcd(q-1, m)=1$.
	Now $\frac{q^m-1}{q-1}=1+q+\cdots+q^{m-1}$ is odd and $2\nmid \frac{q^m-1}{q-1}$.
	If $q$ is even, then $\frac{q^m-1}{q-1}$ is always an odd integer.
	The result follows.

	(2) By Theorem \ref{thm-relations} we  have that $\sum\limits_{i=0}^{\ell-1}k_i=k$. By Eq. (\ref{eqn-ki}),
	\begin{eqnarray*}
		\begin{array}{lll}
		|D|=\frac{q^{m-1}-1}{q-1}=\sum\limits_{i=0}^{\ell-1}k_i&=&\sum\limits_{i=0}^{\ell-1}\frac{1}q\left(n+ \chi(C_{i'}^{(\ell, q^m)})\right) \\
    	&=&\frac{1}q(\ell n+\sum\limits_{i=0}^{\ell-1}\chi(C_{i'}^{(\ell, q^m)}) \\
		&=&\frac{1}q(\frac{q^m-1}{q-1}+\sum\limits_{i=0}^{\ell-1}\chi(C_{i'}^{(\ell, q^m)})),
    	\end{array}
   \end{eqnarray*}
   where $i'=i(q-1) \bmod{\ell}$. It is easy to see that the above equation holds if and only if
   $\sum\limits_{k=0}^{\ell-1}\chi(C_{i'}^{(\ell, q^m)})=-1$.
\end{proof}

We can determine the cardinality $k_i$ for certain cases.
Using Lemma \ref{lem-semipri}, we have the following result.

\begin{theorem}
	\label{semiprimitive}
    Let $r=q^m$.
	Assume that $p$ is a prime, $\ell \ge 2$ is a positive integer, $r=p^{2j\gamma}$,
	where $\ell|(p^j+1)$ for some $j$ and $j$ is the smallest such positive integer. Then the
	cardinalities $k_i$ of the base blocks $D_i$ are given as follows.

	{\bf{Case A.}} If $\gamma, p, \frac{p^j+1}{\ell}$ are all odd, then
	\begin{eqnarray*}
		k_{i'} &=& \frac{r-1 + ((\ell-1)\sqrt{r}-1)(q-1)}{\ell (q-1)q},\ i'(q-1)\equiv \frac{\ell}{2}\bmod \ell, \\
	k_{i} &=& \frac{(r-1)-(1+\sqrt{r})(q-1)}{\ell (q-1)q},\ \mbox{for all}\ i(q-1)\not\equiv\frac{\ell}{2}\bmod \ell.
	\end{eqnarray*}

    {\bf{Case B.}} In all the other cases,
	\begin{eqnarray*}
	k_{0} &=& \frac{r-1 -(q-1)[(-1)^\gamma (\ell-1)\sqrt{r}+1]}{\ell (q-1)q},\\
	k_{i} &=& \frac{r-1 +(q-1)[(-1)^\gamma \sqrt{r}-1]}{\ell (q-1)q}, \ \mbox{for all}\ 0<i<\ell.
	\end{eqnarray*}
\end{theorem}

\begin{proof}
The conclusions follow from (\ref{eqn-ki}) and Lemma \ref{lem-semipri}.
\end{proof}

\begin{remark}
  The determination of $k_i$ for Singer difference sets are also discussed in \cite{kestenband}.
  However, the author of  \cite{kestenband} considered only the case $\ell=q+1$ and used geometric arguments, where $\ell$ is defined in Theorem \ref{semiprimitive}. Recently, we noticed that this problem is also discussed in \cite{momihara}. 
\end{remark}

Next we give an example.

\begin{example}
	Let $q=2, r=q^6,$ and $\ell=3$. Let $g$ be a primitive element of $\gf_r$ and $\alpha=g$.
    By Theorem \ref{thm-new}, the set $$D=\{i: 0\le i< q^6-1 | \tr(\alpha^i)=0 \}$$
	is a $(63,31,15)$ difference set in $\Z_{63}$. By Theorem \ref{thm-mainmain}, we get a
	$(21,\{9,13\},15;3)$ difference family in $\Z_{63}$ with the following base blocks:
	\begin{eqnarray*}
		\begin{array}{lll}
		D_0 &=& \{0, 3, 5, 6, 9, 10, 12, 13, 15, 17, 18, 19, 20\}, \\
	    D_1 &=& \{0, 1, 2, 3, 5, 9, 11, 13, 16\}, \\
	    D_2 &=& \{0, 1, 2, 4, 5, 6, 10, 11, 18\}.
	\end{array}
	\end{eqnarray*}
\end{example}

\subsubsection{The construction with twin-prime difference sets}

In this section, we determine the cardinalities $k_i$ of the blocks of the cyclic difference family obtained from the
the twin-prime difference sets.

\begin{theorem}
	\label{twinds}
Let $q$ and $q+2$ be prime powers. Then the set
$$D=\{ (x,y): x\in\gf_q, y\in\gf_{q+2} | x, y\ \mbox{both are nonzero squares or nonsquares or}\ y=0 \}$$
is a $(q(q+2), \frac{q^2+2q-1}2, \frac{q^2+2q-3}4)$ difference set in $(\gf_q, +)\times (\gf_{q+2}, +)$.
\end{theorem}

Since we only consider cyclic difference sets, we assume $q$ and $q+2$ are both primes in this section.
Let $G=(\gf_q, +)\times (\gf_{q+2}, +)$, it is clear that $G\cong \Z_{q(q+2)}$ as $\gcd(q,q+2)=1$.
Let $$\phi: \Z_{q(q+2)}\rightarrow \Z_q\times\Z_{q+2}$$
defined by $\phi(x)=(x \bmod q, x \bmod (q+2))$. By the Chinese Remainder Theorem, $\phi$ is a group isomorphism.
Let $D$ be the twin-prime difference set in Theorem \ref{twinds}.
As in Section \ref{sec-dsDF}, let $s(t)$ be the characteristic sequence of $D$.
Assume $v=\ell n=q(q+2)$, then $(\ell,n)$ can either be $(q,q+2)$ or $(q+2,q)$.
For $0\le i\le \ell-1$, define the sequence $$s_i(t)=s(\ell t+i),\ 0\le t\le n-1.$$
It can be seen that $1=s_i(t)=s(\ell t+i)$ if and only if $\phi(\ell t+i)\in D$.
Now we give the following result.

\begin{theorem}
	\label{twin-size}
Let notations be as above in this section and $k_i=|D_i|$ for $0\le i\le \ell-1$.

(1) If $(\ell,n)=(q,q+2)$, then
\begin{eqnarray*}
k_i=\left\{ \begin{array}{ll}
	         \frac{q+3}{2} & \mbox{ if }\ i\ne 0  \\
             1           & \mbox{ if }\ i=0.
             \end{array}
     \right.
\end{eqnarray*}

(2) If $(\ell,n)=(q+2,q)$, then
\begin{eqnarray*}
k_i=\left\{ \begin{array}{ll}
             \frac{q-1}2 & \mbox{ if }\ i \ne 0  \\
             q           & \mbox{ if }\ i=0.
             \end{array}
     \right.
\end{eqnarray*}
\end{theorem}

\begin{proof}
We only prove (1) as the proof of (2) is similar.
For $0\le i\le q-1$, we have $1=s_i(t)=s(qt+i)$ if and only if
$$\phi(qt+i)=((qt+i)\ \mbox{bmod}\ q, (qt+i)\ \mbox{bmod}\ (q+2))=(i, (qt+i)\ \mbox{bmod}\ (q+2) )\in D.$$
It is easy to verify that $\{ (qt+i)\ \mbox{bmod}\ (q+2): t\in\gf_{q+2} \}=\gf_{q+2}$.
We only need to prove that $(qt_1+i)\not\equiv (qt_2+i)\ \mbox{bmod}\ (q+2)$ for $t_1,t_2\in\gf_{q+2}, t_1\ne t_2$.
When $0\ne i\in\gf_q$, we have
\begin{eqnarray*}
	\begin{array}{lll}
    k_i&=&|\{ t\in\gf_{q+2} | s_i(t)=1 \}|  \\
       &=&|\{ t\in\gf_{q+2} | (i, (qt+i)\ \mbox{bmod}\ (q+2))\in D \}| \\
	   &=&\frac{q+1}2+1. \\
	\end{array}
\end{eqnarray*}
Similarly, $k_i=1$ when $i=0$.
This completes the proof.
\end{proof}

\begin{remark}
	\label{twin-cor}
	By Theorem \ref{twin-size},   $k_0=1$ when $(\ell,n)=(q,q+2)$. It is clear that we may
	delete $D_0$ and get a $(q+2, \frac{q+3}{2}, \frac{q^2+2q-3}4; \ell-1)$ difference family in $\Z_{q+2}$.

	Similarly, it can be seen that $D_0=\gf_q$ when $(\ell,n)=(q+2,q)$. After deleting $D_0$, we get
	a $(q, \frac{q-1}2, \frac{q^2+2q-3}4; \ell-1)$ difference family in $\Z_q$.
\end{remark}

\begin{example}
Let $q=11$, then
$$D=\{ (x,y):  x, y\ \mbox{both are nonzero squares or nonsquares or}\ y=0 \}$$
is a $(143, 71, 35)$ difference set in $\Z_{143}$.
If $(\ell,n)=(q,q+2)$, by Theorems \ref{thm-mainmain}, \ref{twin-size} and Remark \ref{twin-cor}
we have a $(13, 7, 35; 10)$ difference family in $\Z_{13}$ with the following base blacks:
\begin{eqnarray*}
	\begin{array}{lll}
    &\{ 0, 3, 4, 5, 6, 9, 11 \}, &\{ 1, 2, 4, 5, 6, 8, 9 \}, \\
    &\{ 0, 2, 4, 7, 8, 9, 10 \}, &\{ 0, 1, 2, 4, 5, 10, 11 \}, \\
    &\{ 0, 1, 6, 7, 9, 10, 11\}, &\{ 0, 2, 3, 4, 6, 7, 12 \},  \\
    &\{ 0, 1, 3, 4, 5, 7, 8 \},  &\{ 1, 2, 3, 4, 7, 9, 11 \}, \\
    &\{ 0, 1, 2, 5, 7, 9, 12 \}, &\{ 0, 1, 2, 3, 6, 8, 10 \}. \\
    \end{array}
\end{eqnarray*}

If $(\ell,n)=(q+2,q)$, we have a  $(11, 5, 35; 12))$ difference family in $\Z_{11}$ with the following
base blocks:
\begin{eqnarray*}
	\begin{array}{lllll}
    & \{ 0, 4, 5, 6, 8 \},  &\{ 0, 1, 2, 4, 7 \},  & \{ 0, 1, 2, 4, 7 \},  &\{ 0, 2, 3, 4, 8 \},  \\
    & \{ 0, 1, 5, 8, 10 \}, &\{ 0, 3, 7, 8, 9 \},  & \{ 1, 3, 4, 5, 9 \},  &\{ 0, 1, 2, 6, 9 \},  \\
	& \{ 0, 1, 3, 6, 10 \}, &\{ 1, 2, 3, 5, 8 \},    & \{ 0, 3, 5, 6, 7 \},  &\{ 1, 4, 6, 7, 8 \}.
	\end{array}
\end{eqnarray*}
\end{example}

\section{The construction with disjoint difference families}\label{sec-disj}

The fifth construction of difference families is presented in this section.
Similar to Section 4, we first use near-resolvable block design to obtain more general construction
of difference families and then apply it on disjoint difference families to obtain a specific construction.

A \textit{partial parallel class} is a set of blocks that contain no point of the design more than once.
A \textit{near parallel class} is a partial parallel class missing a single point.
A $(v,k,k-1)$ \textit{near-resolvable block design} (NRB) is a $(v,k,k-1)$-BIBD
with the property that the blocks can be partitioned into near-parallel classes.
For such a design, every point is absent from exactly one class, which implies that
it has $v$ near parallel classes.

\begin{theorem}
  \label{NRBconstruction}
  Let $\mathcal{D}=(\mathcal{V},\mathcal{B},\mathcal{R})$ be a $(kn+1,k,k-1)$-NRB and, for any fixed $s$ between $1$ and $n$, let $\mathcal{R}_s$
  be the set of all possible unions of $s$ pairwise distinct blocks belonging to the same near parallel class of $\mathcal{F}$.
  Then $\mathcal{D}^\cup=(\mathcal{V},\mathcal{R}_s)$ is a $(kn+1,ks,\lambda_s)$-BIBD with 
  $\lambda_s=(k-1)\left( n-1 \atop{s-1} \right) +(kn-k)\left(n-2 \atop{s-2}\right)$.
\end{theorem}
\begin{proof}
  We only need to prove that for any two points $x,y\in\mathcal{V}$, there are exactly $\lambda_s$ blocks containing them.
  Notice that every near parallel class of $\mathcal{R}$ has $n$ blocks.
  Clearly if either $x$ or $y$ are not in a near parallel class $\mathcal{C}$ of $\mathcal{R}$, then
  all possible unions of $s$ pairwise distinct blocks belonging to $\mathcal{C}$ do not contain $x,y$.
  Since $\mathcal{D}$ has $v=kn+1$ near parallel classes and every point is absent from exactly one class,
  there are $v-2=kn-1$ near parallel classes containing both $x,y$.
  Next, we consider the following two cases of the near parallel class $\mathcal{C}$ containing $x,y$ of $\mathcal{R}$:
  $(1)$ $\mathcal{C}$ contains the block $B\in\mathcal{B}$ such that $x,y\in B$.
  Clearly there are $k-1$ such classes and for each such class there are $\left(n-1\atop{s-1}\right)$ choices
  of the other $s-1$ blocks of $\mathcal{D}$.
  So altogether there are $(k-1)\left(n-1 \atop{s-1}\right)$ blocks of $\mathcal{D}^\cup$ containing $x,y$ in this case;
  $(2)$ $\mathcal{C}$ does not contain the block $B\in\mathcal{B}$ such that $x,y\in B$.
  Now it is not difficult to see that there are $(v-2)-(k-1)=(kn-1)-(k-1)=kn-k$ such near parallel classes of $\mathcal{D}$
  and every class has $\left(n-2\atop s-2\right)$ choices of the other $s-2$ blocks of $\mathcal{D}$.
  Therefore, for this case there are $(kn-k)\left(n-2\atop {s-2}\right)$ blocks of $\mathcal{D}^\cup$ containing $x,y$.
  Combining cases $(1)$ and $(2)$, we prove the result.
\end{proof}

Now let $D=\{D_1, D_2, \ldots, D_u\}$ be a $(v,k,k-1)$ disjoint difference family of $G$.
Then the block design $\calD=(G,\mathcal{B})$ is a $(v,k,k-1)$-NRB, where $\mathcal{B}=\{D_ix:0\le i\le u, x\in G\}$.

\begin{corollary}
  \label{construction4}
Let $G$ be a group of order $v$, and let $\calD=\{D_1, D_2, \ldots, D_u\}$
be a collection of $k$-subsets of $G$. For each $1\le s\le u-1$, let $\Omega_s=\{x\in 2^U, \ |x|=s\}$, $\mu=\left( u\atop s\right)$
and $t_i=\left( u-i \atop s-i \right)$ for $i=1,2$,
where $U=\{1,2,\cdots,u\}$ and $2^U$ is the power set of $U$. If $\calD$ is a partition of $G\setminus\{1_G\}$ and
a $(v, k, k-1;u)$-DDF in $G$, then
\begin{eqnarray*}\label{eqn-lastone}
\left\{\bigcup_{i\in x}D_i : x\in \Omega_s \right\}
\end{eqnarray*}
is a $(v, sk, t_1(k-1)+t_2(v-k-1); \mu)$-DF in $G$.
\end{corollary}

An example of $(v,k,k-1)$-DDF is as follows.
Let $q$ be a power of an odd prime, and let $\alpha$ be a generator of $\gf_{q}^*$.
Assume that $q-1=el$, where $e >1$ and $l >1$ are integers. Define
$C_0^{(e)}$ to be the subgroup of $\gf_{q}^*$ generated by $\alpha^e$, and let $C_i^{(e)}
:=\alpha^i C_0^{(e)}$ for each $i$ with $0 \leq i \leq e-1$. These $C_i^{(e)}$
are called {\it cyclotomic classes} of order $e$ with respect to $\gf_{q}^*$.

\begin{lemma}\label{lem-DF1} {\rm (Wilson \cite{Wilson})}
Let $q-1=el$ and let $q$ be a power of an odd prime. Then
$\calD:=\{C_0^{(e)}, \ldots, C_{e-1}^{(e)}\}$ is a $(q, (q-1)/e,
(q-1-e)/e, e)$-DDF in $(\gf_{q},+)$.
\end{lemma}

\begin{corollary}
Let $q-1=el$ and let $q$ be a power of an odd prime. Let
$\calD:=\{C_0^{(e)}, \ldots, C_{e-1}^{(e)}\}$ be the $(q, (q-1)/e,
(q-1-e)/e, e)$-DDF in $(\gf_{q},+)$.
For each $1\le s\le e-1$, let $\Omega_s=\{x\in 2^U, \ |x|=s\}$, $\mu=\left( e\atop s\right)$
and $t_i=\left( e-i \atop s-i \right)$ for $i=1,2$,
where $U=\{1,2,\cdots,e-1\}$ and $2^U$ is the power set of $U$. Then
\begin{eqnarray*}
\left\{\bigcup_{i\in x}D_i : x\in \Omega_s \right\}
\end{eqnarray*}
is a $(q, s(q-1)/e, t_1(q-1-e)/e+t_2(eq-q+1-e)/e; \mu)$-DF in $(\gf_{q},+)$.
\end{corollary}

\begin{proof}
  The conclusion follows from Corollary \ref{construction4}  and Lemma \ref{lem-DF1}.
\end{proof}

\section{The construction with almost difference sets}\label{sec-adscons}

Finally, we describe the construction of difference families using almost difference sets.
Let $D$ be a $(v,k,\lambda,t)$ almost difference set in $G$. Suppose that the differences
$$d_1-d_2, d_1,d_2\in D, d_1\ne d_2$$
represent each element in $T$ exactly $\lambda$ times and
each nonzero element in $G\setminus T$ exactly $\lambda+1$ times.
For each $t\in T$ define the subset
$$\Delta_t=\{d-t: d \in D| d-t\not\in D\},$$
and $\delta_t$ is defined to be an arbitrary element in $\Delta_t$.
Now for each $x\in G^*$,
\begin{eqnarray}
\bar{D}_x=\left\{ \begin{array}{ll}
                  D_x & \mbox{ if } x\not\in T \\
                  D_x \cup \{\delta_x\} & \mbox{ if } x\in T, \\
                  \end{array}
          \right.
\end{eqnarray}
where $D_x=D\cap (D+x)$.

\begin{theorem}\label{adsdf}
The set $\{\bar{D}_x: x\in G^*\}$ is a $(v, \lambda+1, \lambda(\lambda+1); v-1)$
difference family in $G$ if and only if 
for each $t\in T$,
$$|\{x\in G^* | \delta_t, \delta_t+t\in \bar{D_x} \}|=\lambda.$$
\end{theorem}

\begin{proof}
Let $G^*=G\setminus\{0\}$. Since $D$ is a $(v,k,\lambda;t)$ almost difference set,
$|D_x|=\lambda$ for every $x \in T$ and $|D_x|=\lambda+1$ for every $x \in G^*\setminus T$.
By the definition of $\bar{D_x}$ we see that $\bar{D_x}=\lambda+1$ for all $x\in G^*$.

We now compute the multiset $M:=\{ d_1-d_2: d_1,d_2\in \bar{D_x}, x\in G^* | d_1\ne d_2 \}$.
Clearly there are $\lambda(\lambda+1)(v-1)$ differences $d_1-d_2$ in $M$.

Case 1: For each $g\in G^*\setminus T$, there are $\lambda+1$ pairs $(d_1,d_2)\in D\times D$
such that $g=d_1-d_2$. Let $(d_1,d_2)$ be such a pair, note that $D_x\subseteq \bar{D_x}$ for all
$x\in G^*$. Using the same argument in Theorem \ref{dsdf} we see that $d_1-d_2$ appears at least
$(\lambda+1)-1=\lambda$ times in $M$. Hence the multiplicity of $g$ in $M$ is at least
$\lambda(\lambda+1)$. Since $|G^*\setminus T|=v-t-1$, at least $(v-t-1)\lambda(\lambda+1)$ differences
$d_1-d_2$ are used in $M$ in this case.

Case 2: For each $g\in T$, there are $\lambda$ pairs $(d_1,d_2)\in D\times D$
such that $g=d_1-d_2$. Similarly, for each such pair $(d_1,d_2)$, the difference $d_1-d_2$ appears
at least $\lambda$ times.
Now by the definition of $\bar{D_g}$, we have $g=(\delta_g+g)-\delta_g$. Since $\delta_g\not\in D$
and $\delta_g+g=(d-g)+g\in D$ for some $d\in D$, the difference $g=(\delta_g+g)-\delta_g$ is not
included in the above $\lambda$ differences.
Now we compute the multiplicity of $(\delta_g+g)-\delta_g$ in $M$.
The number of $x \in G^*$ such that
$$
\delta_g+g \in \bar{D_x} \mbox{ and } \delta_g \in \bar{D_x}
$$
is equal to $\lambda$ by the assumption.
Hence, the multiplicity of $g$ in $M$ is at least $\lambda+\lambda^2=\lambda(\lambda+1)$
and at least $\lambda(\lambda+1)t$ differences in $M$ are used in this case.

Combining the above two cases, we use at least
$(v-t-1)\lambda(\lambda+1)+\lambda(\lambda+1)t=\lambda(\lambda+1)(v-1)$ differences
with the form in $M$. However, there are $\lambda(\lambda+1)(v-1)$
differences altogether in $M$. Therefore, the multiplicity of each element in $G^*$ is $\lambda(\lambda+1)$
and we prove the necessary part.

Conversely, if $\{\bar{D}_x: x\in G^*\}$ is a $(v, \lambda+1, \lambda(\lambda+1); v-1)$ difference family,
the conclusion follows easily.

\end{proof}

Let $q \equiv 1 \bmod{4}$ be a power of a prime. Let $D$ be the set of all quadratic
residues in $(\gf_{q}, +)$.
It is known that $D$ is a $(q, (q-1)/2, (q-5)/4; (q-1)/2)$ almost difference set in $(\gf_{q}, +)$.

For each $x \in \gf_{q}^*$, define
$$
D_x=D \cap (D +x),
$$
where $D+x=\{d+x: d \in D\}$. It follows from the cyclotomic numbers of order two that
$|D_x|=(q-5)/4$ if $x$ is a quadratic residue in $\gf_{q}^*$ and $|D_x|=(q-1)/4$ if $x$
is a quadratic nonresidue in $\gf_{q}^*$ \cite{Stor}.

We now define for each $x \in \gf_{q}$
\begin{eqnarray}
\bar{D}_x=\left\{ \begin{array}{ll}
                  D_x & \mbox{ if $x$ is not a square} \\
                  D_x \cup \{0\} & \mbox{ if $x$ is a square.}
                  \end{array}
          \right.
\end{eqnarray}

\begin{theorem}\label{thm-lll}
Then $\{\bar{D}_x: x \in \gf_{q}^*\}$ is a $(q, (q-1)/4, (q-1))(q-5)/16); q-1)$
difference family in $(\gf_{q}, +)$.
\end{theorem}

\begin{proof}
In the proof of Theorem \ref{adsdf}, let $T=D$ and let $\delta_x=0$ for all $x\in D$.
The condition $|\{x\in G^* | 0, t\in \bar{D_x} \}|=\frac{p-5}4=\lambda$ is clearly satisfied.
Then the conclusion of this theorem follows from that of Theorem \ref{adsdf}.
\end{proof}

\begin{example}
Let $q=13$. Then the following blocks form a $(13, 3, 6; 12)$ difference family in
$(\gf_{13},+)$:
\begin{eqnarray*}
&& \{ 0, 4, 10 \},
\{ 1, 3, 12 \},
\{ 0, 4, 12 \},
\{ 0, 1, 3 \},
\{ 1, 4, 9 \},
\{ 3, 9, 10 \},
\\
&& \{ 3, 4, 10 \},
\{ 4, 9, 12 \},
\{ 0, 10, 12 \},
\{ 0, 1, 9 \},
\{ 1, 10, 12 \},
\{ 0, 3, 9 \}.
\end{eqnarray*}
\end{example}

\section{Concluding remarks}\label{sec-final} 

In this paper, we presented six constructions of difference families. Four of them employ difference sets, 
one uses disjoint difference families and the other one uses almost difference sets. All constructions do not require 
that $G$ be Abelian, except the forth one. The third and fifth constructions yield more general constructions 
of difference families using BIBDs and NRBs. These new constructions yield many new parameters of 
difference families. Furthermore, they established new bridges between these combinatorial objects and 
difference families.

As a summary, Table \ref{tab-222222} lists all parameters of difference families
obtained in this paper. Unlikely Table \ref{tab-111111}, we only give the parameters from the generic constructions,
i.e., using the parameters of the difference sets, disjoint difference families, or almost difference sets.
In Table \ref{tab-222222}, the symbol $\star$ denotes the constructions use difference sets, $\clubsuit$ denotes the ones use
disjoint difference families and $\spadesuit$ denotes the ones use almost difference sets.
The parameters $v,k,\lambda$ are the ones of the corresponding combinatorial objects.

{\footnotesize{
\begin{table}[ht]
   \centering
   \caption{Summary of parameters of difference families in this paper, where $k$ is the block size, $\gamma$ is the frequency}\label{tab-222222}
   \begin{tabular}{|c|c|c|c|c|c|} \hline
    $v$ & $k$       & $\gamma$               & Conditions    & reference & \\ \hline \hline
    $v$ & $\lambda$ & $\lambda(\lambda-1)$   &               & Thm. \ref{dsdf}  &$\star$\\ \hline
    $v$ & $\lambda$ & $\lambda(\lambda-1)/2$ &$\begin{array}{lll}
                                                 && v\  \mbox{odd and}\ \gcd(v,\lambda)=1,\ \mbox{or} \\
						 && D\ \mbox{is a skew Hadamard}\\
						 && \mbox{difference set} \\
					       \end{array}$                 & Thm. \ref{construction2} &$\star$\\ \hline
    $v$ & $k+s$ &$\begin{array}{lcc}
                    &&\lambda\left(v-k \atop{s}  \right)+2(k-\lambda)\left(v-k-1 \atop{s-1}\right)\\
		    &&+(v-2k+\lambda)\left( v-k-2 \atop{s-2} \right) \\
		  \end{array}$  &$1\le s\le v-k-1$ &Cor. \ref{construction3} &$\star$\\ \hline
    $v$ & $k-s$ &$\lambda\left(k-2 \atop{s-2} \right)$ &$1\le s\le k-1$ & Cor. \ref{construction5} &$\star$ \\ \hline
    $n$ & $K$ & $\lambda$ &$\begin{array}{lll}
                             && G\ \mbox{is Abelian},\ \forall n\mid v \\
			     && K\ \mbox{is the set of the block sizes}\\
			     && \mbox{depends on }\ D \\
			   \end{array}$			     &Thm. \ref{thm-mainmain} &$\star$\\ \hline
    $v$ & $sk$ & $(k-1)\left( n-1 \atop{s-1} \right) +(kn-k)\left(n-2 \atop{s-2}\right)$  &$D$ is a $(v,k,k-1)$-DDF in $G$ &Cor. \ref{construction4} &$\clubsuit$ \\ \hline
    $v$ & $\lambda$ &$\lambda(\lambda+1)$ &$\begin{array}{lll}
			                                            && D\ \mbox{is}\ (v,k,\lambda,t)\mbox{-ADS in}\ G, \\ 
                                                                    && |\{x\in G^* | \delta_t, \delta_t+t\in \bar{D_x} \}|=\lambda \\
								    && \mbox{for each}\ t\in T, \\
			                                        \end{array}$  &Thm. \ref{adsdf} &$\spadesuit$\\ \hline
			 \end{tabular}
\end{table}
}}


\begin{thebibliography}{99}

\bibitem{ADHKM} K. T. Arasu, C. Ding, T.Helleseth, P. V. Kumer, and H. Martinsen,
Almost difference sets and their sequences with optimal autocorrelation, IEEE Trans. Inform.
Theory 47 (2001), 2834--2843.

\bibitem{Baum} L. D. Baumert. Cyclic Difference Sets, volume 182 of Lecture Notes in Mathematics.
Springer-Verlag, 1971.

\bibitem{baumert} L. D. Baumert, M. H. Mills, and R. L. Ward, Uniform Cyclotomy, J. Number Theory 14 (1982), 67--82.

\bibitem{BJL} T. Beth, D. Jungnickel and H. Lenz, Design Theory,
Cambridge University Press, Cambridge, 1999.

\bibitem{Bura} M. Buratti, Constructions of $(q,k,1)$ difference families with $q$
a prime power and $k=4,5$, Discrete Mathematics 138 (1995), 169--175.

\bibitem{CD2004} C. Carlet and C. Ding, Highly nonlinear functions, J. Complexity 20 (2004) 205--244.

\bibitem{CD05} Y. Chang and C. Ding, Constructions of external difference families and
disjoint difference families, Designs, Codes and Cryptography 40(2) (2006), 167--185.

\bibitem{Chen} Y. Q. Chen, A construction of difference sets, Designs Codes and Cryptography
13 (1998), 247--250.

\bibitem{CZ}
K. Chen and L. Zhu, Existence of $(q,k,1)$ difference families with $q$ a
prime power and $k=4,5$, J. Combin. Designs 7 (1999), 21--30.

\bibitem{Chol} S. Chowla, A property of biquadratic residues, Proc. Nat. Acd. Sci.
India A 14 (1944), 45--46.

\bibitem{DY04} C. Ding and J. Yuan, A family of skew difference sets,
J. Combinatorial Theory Ser A 113(7) (2006), 1526--1535.

\bibitem{Cding} C. Ding, Two constructions of $(v, (v-1)/2, (v-3)/2)$ difference families,
      J. Combinatorial Designs 16(2) (2008), 164--171.

\bibitem{DWX} C. Ding, Z. Wang and Q. Xiang, Skew Hadamard difference sets from the Ree-Tits slice
symplectic spreads in PG$(3,3^{2h+1})$ , J. Combin. Theory Ser A 114 (5), (2007), 867--887.


\bibitem{DR}
J. H. Dinitz and P. Rbmodney, Disjoint difference families with block size 3,
Utilitas Math. 52 (1997), 153--160.

\bibitem{DS}
J. H. Dinitz and N. Shalaby, Block disjoint difference families for Steiner triple
systems: $v\equiv 1\ ({\rm bmod }\ 6)$, J. Statist. Plann. Inference 106 (2002), 77--86.


\bibitem{feng-skds}
T. Feng, 
Non-abelian skew Hadamard difference sets fixed by a prescribed automorphism,
J. Combin. Theory Ser. A 118 (2011), no. 1, 27--36. 

\bibitem{Miao}
R. Fuji-Hara, Y. Miao and S. Shinohara, Complete sets of of disjoint difference families
and their applications, J. Statist. Plann. Inference 106 (2002), 87--103.

\bibitem{kestenband}
B. C. Kestenband, Addentum to James Singer's theorem on difference sets,
European J. Combin. 25 (2004), 1123--1133. 

\bibitem{Lehm} E. Lehmer, On residue difference sets, Canad. J. Math. 5 (1953), 425-432.

\bibitem{lll}
C. Li, Q. Li and S. Ling, Properties and applications of preimage distributions of
perfect nonlinear functions, IEEE Trans. Inform. Theory 55 (2009), no 1,  64--69.


\bibitem{Mart} L. Martinez, D. Z. Dokovic, A. Vera-Lopez, Existence question for
difference families and construction of some new families, J. Combinatorial Designs
12(4) (2004), 256--270.

\bibitem{momihara}
  K. Momihara, New optimal optical orthogonal codes by restrictions to subgroups,
  Finite Fields Appl. 17 (2011), no. 2, 166--182.


\bibitem{nair} U. P. Nair, Elementary results on the binary quadratic form $a^2+ab+b^2$,
2004 [online]. Available: \url{http://arxiv.org/abs/math/0408107}.

\bibitem{Niven} I. Niven, H. S. Zuckerman, H.L. Montgomery, An introduction to the
theory of numbers, Fifth Edition, John Wiley, 1991.

\bibitem{Ogata} W. Ogata, K. Kurosawa, D. R. Stinson, H. Saido, New combinatorial
designs and their applications to authentication codes and secret sharing schemes,
Discrete Mathematics 279 (2004), 383--405.



\bibitem{Pale} R. E. A. C. Paley, On orthogonal matrices, J. Math. Physics MIT
12 (1933), 311--320.

\bibitem{Pott} A. Pott, Finite Geometry and Character Theory, Springer Press, 1995.

\bibitem{Sing} J. Singer, A theorem in finite projective geometry and some
applications to number theory, Trans. AMS 43 (1938), 377--385.

\bibitem{Stor} T. Storer, Cyclotomy and Difference Sets, Markham, Chicago, 1967.


\bibitem{weng}
G. B. Weng, W. S. Qiu, Z. Y. Wang, Q. Xiang,
Pseudo-Paley graphs and skew Hadamard difference sets from presemifields,
Des. Codes Cryptogr. 44 (2007), no. 1-3, 49--62.



\bibitem{Whit} A. L. Whiteman, A family of difference sets, Inninois J. Math. 6 (1962), 53--76.

\bibitem{Wilson} R. M. Wilson, Cyclotomy and difference families in elementary
      Abelian groups, J. Number Theory 4 (1972), 17--42.

\end{thebibliography}
\end{document}